\documentclass[journal,twosided,web]{ieeecolor}
\usepackage{generic}
\usepackage{cite}
\usepackage{amsmath,amssymb,amsfonts}
\usepackage{algorithmic}
\usepackage{graphicx}
\usepackage{textcomp}
\usepackage{dsfont}
\usepackage{color}
\usepackage{epstopdf}        
\usepackage{enumerate}
\usepackage{soul}
\usepackage{lscape}
\usepackage{multicol}
\usepackage{multirow}
\usepackage{calligra}
\usepackage{booktabs}
\usepackage{mathtools}
\usepackage{framed} 
\usepackage{picins}
\usepackage{empheq}
\usepackage{afterpage}
\usepackage{graphics} % for pdf, bitmapped graphics files
\usepackage{epsfig} % for postscript graphics files
\usepackage{times} % assumes new font selection scheme installed
\usepackage{amsmath} % assumes amsmath package installed
\usepackage{amssymb}  % assumes amsmath package installed
\usepackage{amsfonts}  % assumes amsmath package installed
\usepackage{subfig}
\usepackage{epstopdf}
\usepackage{enumerate}
\usepackage{graphicx} % grÂficos
\usepackage{xcolor}
\usepackage{comment} 

\newcommand{\tcb}{\textcolor{black}}

\DeclareMathOperator*{\argmin}{arg\,min}
\newtheorem{thm}{Theorem}
\newtheorem{prop}{Proposition}
\newtheorem{lemma}{Lemma}

\newtheorem{definition}{Definition}
\newtheorem{assumption}{Assumption}
\newtheorem{remark}{Remark}
\newtheorem{example}{Example}

\newcommand\tx{\tilde{x}}
\newcommand\tu{\tilde{u}}
\newcommand\bj{\mathbf{J}}
\newcommand\re{\mathbb{R}}

\newcommand\cj{\mathcal{J}}
\newcommand\cc{\mathcal{C}}

\newcommand\cg{\mathcal{G}}

\newcommand\ck{\mathcal{K}}

\newcommand\vpt{\varphi(\theta)}
\newcommand\cki{\mathcal{K}_\infty^{-P}}
\newcommand\ckf{\mathcal{K}_\infty^{\text{FxT}}}
\newcommand\fx{\mathcal{F}_{\xi}}
\newcommand{\ov}{\overline}
\newcommand{\un}{\underline}
\DeclarePairedDelimiter\sg{\lfloor}{\rceil}
\newcommand\ckp{\mathcal{K}_\infty^P}
\newcommand\klfx{\mathcal{KL}_{\text{FxT}}}

\newcommand*{\QEDB}{\hfill\ensuremath{\square}}%

\def\BibTeX{{\rm B\kern-.05em{\sc i\kern-.025em b}\kern-.08em
    T\kern-.1667em\lower.7ex\hbox{E}\kern-.125emX}}
\begin{document}
\title{A Lyapunov-Based Small-Gain Theorem for Fixed-Time ISS:\\Theory, Optimization, and Games}

\author{Michael Tang, Miroslav Krstic, Jorge I. Poveda
%\thanks{This paragraph of the first footnote will contain the date on 
%which you submitted your paper for review. It will also contain support 
%information, including sponsor and financial support acknowledgment. For 
%example, ``This work was supported in part by the U.S. Department of 
%Commerce under Grant BS123456.'' }
\thanks{M. Tang, J. I. Poveda are with the Department of Electrical and Computer Engineering, University of California, San Diego, La Jolla, CA, USA. Email: myt001@ucsd.edu, poveda@ucsd.edu.}
\thanks{ M. Krsti\'c is with the Department of Mechanical and Aerospace Engineering, University of California, San Diego, La Jolla, CA, USA.}
\thanks{This work was supported in part by NSF grants ECCS CAREER 2305756, CMMI 2228791, and AFOSR FA9550-22-1-0211.}
\vspace{-0.3cm}
}
\maketitle

\begin{abstract}
We develop a Lyapunov-based small-gain theorem for establishing fixed-time input-to-state stability (FxT-ISS) guarantees in interconnected nonlinear dynamical systems. The proposed framework considers interconnections in which each subsystem admits a FxT-ISS Lyapunov function, providing robustness with respect to external inputs. We show that, under an appropriate nonlinear small-gain condition, the overall interconnected system inherits the FxT-ISS property. In this sense, the proposed result complements existing Lyapunov-based small-gain theorems for asymptotic and finite-time stability, and enables a systematic analysis of interconnection structures exhibiting fixed-time stability. To illustrate the applicability of the theory, we study feedback-based optimization problems with time-varying cost functions, and Nash-equilibrium seeking for noncooperative games with nonlinear dynamical plants in the loop. Both settings correspond to model-based equilibrium seeking problems, for which we propose a class of nonsmooth gradient- or pseudo-gradient-based controllers that achieve fixed-time convergence using real-time feedback, \emph{without requiring time-scale separation or passivity-based assumptions}. Numerical examples are provided to validate the theoretical findings.
\end{abstract}

\vspace{-0.3cm}
\section{INTRODUCTION}
\subsection{Motivation}
In recent years, the control and learning communities have devoted considerable effort to the development of algorithms that exhibit faster convergence, higher accuracy, and stronger robustness guarantees. Within this line of research, finite-time stability (FTS) \cite{bhatfinite} and finite-time control methods \cite{zpjfinite} have received significant attention due to their ability to guarantee convergence in finite time and to enhance disturbance rejection in nonlinear systems. A fundamental limitation of finite-time approaches, however, is that the settling time—the time required for trajectories to reach the desired equilibrium—generally depends on the initial conditions and may grow unbounded. This drawback motivated the introduction of fixed-time stability (FxTS) \cite{doi:10.1137/060675861}, a stronger notion of stability that guarantees convergence within a uniform time bound that is independent of the system’s initialization. The main goal of this paper is to introduce a novel analytical tool for the study of fixed-time stability properties in interconnected systems via a small gain theorem. 

\vspace{-0.2cm}
\subsection{Literature Review}
Following the introduction of Lyapunov-based conditions for verifying FxTS in continuous-time autonomous systems \cite{6104367}, the concept has been successfully applied across a broad range of engineering domains, including feedback control \cite{9696292}, network synchronization \cite{hu2020fixed}, optimization dynamics \cite{10885949}, learning algorithms \cite{9760031}, and neural networks \cite{CHEN2020412}, among others. Despite this growing body of work, the analysis of FxTS properties in interconnected systems remains comparatively underdeveloped. Existing methods typically rely on restrictive assumptions, such as specific structural interconnection properties \cite{9714166}, quadratic-type coupling conditions combined with time-scale separation \cite{10644358,11239427}, or trajectory-based small-gain conditions \cite{10886634}. These requirements substantially limit the applicability of the available results to large-scale or heterogeneous interconnected systems.

Small-gain theory has long been a central tool for analyzing the stability of interconnected systems. Originally developed for systems with linear gains \cite{1098316}, it was later extended to nonlinear gains \cite{jiang1994small} after the development of input-to-state stability (ISS) \cite{sontag2008input} theory. These foundational results have since been generalized to a wide variety of system classes, including hybrid systems \cite{bao2018ios,liberzon2012small}, stochastic systems \cite{dragan1997small}, discrete-time systems \cite{jiang2004nonlinear,zhongping2008nonlinear}, and infinite-dimensional systems \cite{mironchenko2021nonlinear,mironchenko2021small}. In parallel, substantial effort has been devoted to reformulating nonlinear small-gain conditions in a Lyapunov framework, enabling constructive stability certificates and facilitating controller synthesis \cite{jiang1996lyapunov,liberzon2014lyapunov,kawan2020lyapunov}.

Despite these extensive developments, extensions of small-gain theory to systems exhibiting accelerated convergence have only begun to emerge in recent years. Trajectory-based small-gain conditions for interconnected finite-time input-to-state stable systems were first derived in \cite{zpjfinite}, and a Lyapunov-based reformulation for networked finite-time systems was proposed in \cite{sgfinite}. More recently, the introduction of fixed-time input-to-state stability (FxT-ISS) \cite{LOPEZRAMIREZ2020104775,fxtiss_converse} enabled the derivation of a trajectory-based small-gain theorem for interconnected FxT-ISS systems \cite{10886634}. However, as acknowledged in \cite{10886634}, these results are conservative and impose additional technical conditions involving the settling-time function, its inverse, and several auxiliary bounding functions, which can restrict the class of admissible nonlinear gains and limit applicability. As a result, Lyapunov-based small-gain tools, similar to those developed for asymptotic \cite{jiang1996lyapunov,liberzon2014lyapunov,kawan2020lyapunov} and finite-time stability \cite{sgfinite}, but capable of certifying FxT-ISS for general interconnected systems, remain largely unexplored.

\vspace{-0.2cm}
\subsection{Contributions}
The main contribution of this paper is the introduction of a Lyapunov-based small-gain theorem to establish FxT-ISS for interconnected systems, where each subsystem is FxT-ISS on its own. This small-gain theorem is instrumental for the design and study of different feedback mechanisms for decision-making in dynamical systems without the use of timescale separation, in particular, feedback optimization and Nash equilibrium-seeking. Specifically, the contributions of this article are as follows:
\begin{enumerate}
\item We introduce a Lyapunov-based small-gain theorem for FxT-ISS that parallels similar tools developed in the literature for the study of asymptotic \cite{jiang1996lyapunov,liberzon2014lyapunov,kawan2020lyapunov} and finite-time \cite{sgfinite} stability. Specifically, we show that if each subsystem admits a FxT-ISS Lyapunov function, where an external input and the state of the other subsystems are treated as the subsystem's input, and the gain functions ---under a mild structural requirement--- satisfy the nonlinear small-gain condition, then the overall interconnection is FxT-ISS. %In this sense, our main result mirrors similar Lyapunov-based small gain theorems developed in the literature for the study of asymptotic stability and finite-time stability. Moreover, we extend the results of \cite{11615748} by considering interconnections with external inputs and further generalizing the class of permissible gain functions.
Our results, presented in Theorem \ref{thm_sgt}, do not follow as simple extensions of previous results, but rather require the derivation of several new technical lemmas, which {establish the
preservation of FxT-ISS Lyapunov inequalities under nonlinear
gain-dependent scalings. These auxiliary results} can be of independent interest for future studies of fixed-time control and FxTS. Academic examples, including the interconnection of two systems with arbitrarily many homogeneous disturbance terms, are presented to illustrate the main result.
\item We investigate the problem of feedback-based optimization \cite{9075378,8673636,bianchin2022online,9540998} with fixed-time convergence guarantees and \emph{without requiring time-scale separation or passivity-based assumptions}. By leveraging the proposed fixed-time small-gain theorem, we show in Theorem \ref{thm_fbk} that a fixed-time gradient flow with unit gain, interconnected in feedback with a plant that has been pre-stabilized in fixed time, drives the trajectories of the resulting closed-loop system toward the solution set of a time-varying optimization problem, provided the curvature of the quasi-steady state cost function is sufficiently large relative to the ``smoothness'' of the dynamics. Since no internal model of the exosystem generating the time variations is assumed to be available, convergence is characterized in terms of a residual tracking error. The error is quantified using the FxT-ISS property, where the input corresponds to the rate of change of the optimizer. {This result complements singular perturbation and small-gain approaches to feedback optimization by providing fixed-time tracking guarantees for time-varying optimizers.}
\item To further demonstrate the applicability of the proposed fixed-time small-gain theorem to multi-agent systems, we extend the feedback optimization framework to Nash equilibrium seeking in non-cooperative games with dynamic plants in the loop \cite{romano2025game,belgioioso2024online,11239427}. We show in Theorem \ref{thm_nes} that a class of decoupled fixed-time pseudo-gradient flows, akin to those studied in \cite{9683248} for static mappings, achieves fixed-time stability with respect to the Nash equilibrium in potential games, even when the game dynamics are interconnected with nonlinear plant dynamics \emph{using no timescale separation or passivity-based assumptions}. To the best of our knowledge, this result is the first in the context of Nash seeking with plants in the loop that achieves fixed-time stability for the closed-loop system.
\end{enumerate}

\vspace{-0.2cm}
\subsection{Additional Contributions with Respect to \cite{11615748}}
Earlier, partial
results of this article were submitted to the 2026 American Control Conference \cite{11615748}\footnote{Submitted in the supplemental material.}. The results of \cite{11615748} consider only systems \emph{without} external inputs {and study closed-loop FxTS}. In contrast, the results of this article pertain to the study of \emph{fixed-time input-to-state stability} via small gain theory, further generalize the class of permissible gain functions compared to those studied in \cite{11615748}, and additionally present two novel applications using fixed-time control that are contributions on their own: the solution of feedback optimization problems with plants in the loop and time-varying cost functions, and the solution of Nash-equilibrium seeking problems in noncooperative games with dynamic plants. The article also presents new numerical illustrative examples, with the complete stability analysis and proofs.

%A Lyapunov-based small-gain theorem for certifying fixed-time stability in interconnected systems was recently developed in \cite{tang2024fixed}, under the assumption that each subsystem is FxT-ISS with respect to the state of the other subsystem. While this framework represents an important step toward Lyapunov-based FxTS analysis of interconnected dynamics, it focuses exclusively on fixed-time stability of autonomous interconnections and does not directly address scenarios in which the interconnection is subject to external inputs or disturbances

The rest of this paper is organized as follows: Section \ref{sec_preliminaries} presents some mathematical preliminaries. Section \ref{sec_main} presents the Lyapunov-based small gain theorem for fixed-time stability. Section \ref{sec_fbkopt} {applies the results to study} the problem of feedback optimization in fixed-time, Section \ref{sec_fixedtime} {similarly considers} the fixed-time Nash equilibrium seeking problem, and Section \ref{sec_conclusions} ends with the conclusions. 

\section{PRELIMINARIES}
\label{sec_preliminaries}
\subsection{Notation}
We use $\re_{\ge 0}$ to denote the set of nonnegative real numbers. Given $n\in\mathbb{N}$, we denote $[n]:=\{1,2,...,n\}$. We use $\cc(X,Y)$ to denote the set of all continuous mappings $f:X\to Y$ between metric spaces $X$ and $Y$. We use $f^{-1}$ to denote the inverse of a function $f$. Moreover, if $X\subset\re$ and $Y\subset\re$, we use $f'$ to denote the derivative of $f$. We introduce the following classes of comparison functions, which will play an important role in this paper:
\begin{align*}
    \ck&:=\{\alpha\in \cc(\re_{\ge 0}, \re_{\ge 0}) : \alpha(0)=0,\ \alpha \text{ strictly increasing}\}\\
    \ck_\infty&:=\{\alpha\in\ck : \lim_{s\to\infty}\alpha(s)=\infty\}\\
    \ckp&:=\{\alpha\in\ck_\infty : \alpha(s)=\sum_{i=1}^n c_i s^{p_i} \text{ for some } c_1,...,c_n\neq 0,\\&~~~~~~~~~~~~~~~~~~ p_1,...,p_n>0, \text{ and } \alpha'(s)>0 \ \forall s>0\}\\
    \ckf&:=\{\alpha\in\ckp : \alpha(s)=c_1 s^{p_1}+c_2 s^{p_2} \text{ for some } c_1, c_2>0,\\&~~~~~~~~~~~~~~~~~~ p_1\in (0,1),\ p_2>1\}\\
    \cki&:=\{\alpha\in\ck_\infty : \alpha^{-1}\in\ckp\}
\end{align*}
% A continuous function $\alpha:\re_{\ge 0}\to\re_{\ge 0}$ is said to be of class $\ck$, denoted $\alpha\in\ck$, if $\alpha(0)=0$ and $\alpha$ is strictly increasing. We use $\ck_\infty$ to denote the set of $\alpha\in\ck$ that satisfy $\lim_{s\to\infty}\alpha(s)=\infty$. Given $\alpha\in\ck_\infty$, if there exists $n\in\mathbb{N}$ and constants $c_i>0$, $p_i>0$ for $i=1,...,n$ such that $\alpha(s)=\sum_{i=1}^n c_i s^{p_i}$ for $s\ge 0$, {and $\alpha'(s)>0$ for each $s>0$,} we say $\alpha$ is of class $\ckp$, denoted $\alpha\in\ckp$. Moreover, if $n=2$  with $p_1\in (0,1)$ and $p_2>1$, we say $\alpha\in\ckf$. Given $\alpha\in\ck_\infty$, we use $\alpha^{-1}$ to denote its inverse. We define $\cki$ to be the class of functions $\alpha\in\ck_\infty$ that satisfy $\alpha^{-1}\in\ckp$. 
A continuous function $\beta:\re_{\ge 0}\times\re_{\ge 0}\to\re_{\ge 0}$ is said to be of class $\klfx$, denoted $\beta\in\klfx$, if $\beta(\cdot, 0)\in\mathcal{K}$ and for each fixed $r\ge 0$, $\beta(r, \cdot)$ is continuous, non-increasing and there exists a continuous function $T:\mathbb{R}_{\ge 0}\to\mathbb{R}_{\ge 0}$ such that $\beta(r, t)=0$ for all $t\ge T(r)$, where $T(0)=0$ and $T$ is uniformly bounded. The mapping $T$ is called the \emph{settling time function}. Given a measurable function $u:\mathbb{R}_{\ge 0}\to\mathbb{R}^m$ we denote $|u|_\infty=\text{ess}\sup_{t\ge 0}|u(t)|$, where $|\cdot|$ represents the Euclidean norm. We use $\mathcal{L}_\infty^m$ to denote the set of measurable functions $u:\mathbb{R}_{\ge 0}\to\mathbb{R}^m$ satisfying $|u|_\infty<\infty$. Given a continuously differentiable function $t\mapsto x(t)$ and a continuous function $x\mapsto V(x)$, the upper-right Dini derivative at $t$ of the function $W(t)=V(x(t))$ is defined as $D^+W(t)=\limsup_{h\to 0^+}\frac{W(t+h)-W(t)}{h}$. Given a continuous function $V:\mathbb{R}^n\to\mathbb{R}_{\geq0}$, the upper right directional Dini derivative  of $V$ at a point $x\in\re^n$ in the direction $v\in\re^n$ is defined to be
\begin{equation*}
    DV(x;v)=\limsup_{h\to 0^+}\frac{V(x+hv)-V(x)}{h}.
\end{equation*}
Given a continuous function $f:\mathbb{R}^n\to\mathbb{R}^m$, we say $f$ is $\ell$-Lipschitz if $|f(x)-f(y)|\le \ell|x-y|$ for all $x,y\in\re^n$. 
% If, for every $x\in\re^n$, where exists a neighborhood $U_x$ of $x$ such that $|f(y)-f(z)|\le L_x|y-z|$ for all $y,z\in U_x$, then we say that $f$ is locally Lipschitz. 
If $f$ is differentiable, we use $\bj_f(x)\in\mathbb{R}^{m\times n}$ to denote the Jacobian of $f$ evaluated at $x\in\mathbb{R}^n$. If $m=1$, we use $\nabla f(x)=\bj_f(x)^\top$. If $\bj_f(x)$ is continuous, we say $f$ is $\mathcal{C}^1$.
\subsection{Auxiliary Lemmas}
The following auxiliary Lemmas will be instrumental for our results:
\begin{lemma}\label{lem_sandw}
    Given $\underline{p}\le p\le \overline{p}$, the following holds
    \begin{equation*}
        x^p\le x^{\underline{p}}+x^{\overline{p}},
    \end{equation*}
    for all $x\ge 0$.
    \QEDB
\end{lemma}
\begin{proof}
    For $x\in [0,1]$ we have $x^p\le x^{\underline{p}}$, and for $x\ge 1$ we have $x^p\le x^{\overline{p}}$. We combine both cases to establish the result.
\end{proof}
\vspace{0.1cm}
\begin{lemma}\label{lem_fxtbd}
    If $\alpha_1,...,\alpha_N\in\ckf$, there exists $\alpha\in\ckf$ such that $\min_k\alpha_k(s)\ge \alpha(s)$ for all $s\ge 0$.\QEDB
\end{lemma}
\begin{proof}
    Let $\alpha_i(s)=a_i s^{p_i}+b_i s^{q_i}$, where $p_i\in (0,1)$ and $q_i>1$ for $i=1,...,N$. Denote $c_i=\min\{a_i, b_i\}$ and $c=\min_k c_k$. By Lemma \ref{lem_sandw}, we have that $\alpha_i(s)\ge c s^r$ for all $s\ge0$, $i\in [N]$, and $r\in\cap_{j=1}^N [p_j, q_j]=[\max_k p_k, \min_k q_k]$. Fix $p\in [\max_k p_k, 1)$ and $q\in (1, \min_k q_k]$, then
    \begin{equation*}
        \min_k\alpha_k(s)\ge \frac{c}{2}(s^p+s^q){:=\alpha(s),}
    \end{equation*}
    {where $\alpha\in\ckf$. }
    This establishes the result. 
\end{proof}
\vspace{0.1cm}
    \begin{lemma}\label{jensenlemma}
    Let $s_i\ge 0$ for each $i\in \{1,2,...,n\}$. If $p\in (0,1]$, then $(\sum_{i=1}^n s_i)^p\le \sum_{i=1}^n s_i^p$. If $p>1$, then $(\sum_{i=1}^n s_i)^p\le n^{p-1}\sum_{i=1}^n s_i^p$.\QEDB
\end{lemma}
\begin{proof}
    The proof is presented in \cite{9760031}.
\end{proof}

\subsection{Fixed-Time Stability and Input-to-State Stability}
Consider the following system
\begin{equation}\label{sysfxts}
    \dot{x}=f(x,u),\quad x(0)=x_0,
\end{equation}
where $x\in\re^n$ is the state, $u\in\mathcal{L}_\infty^m$ is the input, and $f:\re^n\times\re^m\to\re^n$ is a continuous function that satisfies $f(0,0)=0$. Given $x_0\in\re^n$ and $u\in\mathcal{L}_\infty^m$, a solution to \eqref{sysfxts} is represented as $x(t,x_0,u)$ for $t\ge 0$, where $x(0, x_0,u)=x_0$. If $x_0$ and $u$ are clear from the context, we will use $x(t):=x(t, x_0, u)$. When $x$ is defined for all $t\geq0$, we say that $x$ is forward complete. 

In this paper, we are interested in studying fixed-time (ISS) stability properties for system \eqref{sysfxts}. The following definitions, borrowed from \cite{6104367,LOPEZRAMIREZ2020104775,fxtiss_converse}, formalize these notions: 

\vspace{0.1cm}
\begin{definition}
    The origin of \eqref{sysfxts} with $u=0$ is said to be \emph{uniformly globally fixed-time stable (FxTS)} if there exists $\beta\in\klfx$ such that for each $x_0\in\re^n$, every solution of \eqref{sysfxts} is forward complete and satisfies
    \begin{equation}
    |x(t)|\le \beta(|x_0|, t),~~~\forall~~t\geq0.
    \end{equation}
    %
    %.\QEDB
\end{definition}
\vspace{0.2cm}

\begin{definition}
    The origin of \eqref{sysfxts} is said to be \emph{fixed-time input-to-state stable (FxT-ISS)} if there exists $\beta\in\klfx$ and $\varrho\in\mathcal{K}$ such that for each $x_0\in\mathbb{R}^n$ and $u\in\mathcal{L}_\infty^m$, every solution $x(t)$ of \eqref{sysfxts} is forward complete and satisfies
     \begin{equation*}
         |x(t)|\le \beta(|x_0|, t)+\varrho(|u|_\infty),~~\forall~t\geq0.
     \end{equation*}
\end{definition}
\vspace{0.2cm}

The FxT and FxT-ISS properties of \eqref{sysfxts} can be studied via Lyapunov functions. In particular, 
given a locally Lipschitz continuous function $V:\re^n\to\re_{\geq0}$, %the time derivative of $V$ along the solutions to \eqref{sysfxts} is defined as
%
%\begin{equation*}
%    \dot{V}(t):=D^+V(x(t))
%\end{equation*}
%
%for any $t\ge 0$.
we define
\begin{equation*}
    \dot{V}(x,u):=DV(x; f(x,u)).
\end{equation*}
%
% to denote the time derivative  of $V$ along the solutions of \eqref{sysfxts}, \tcr{which is defined almost everywhere due to Rademacher’s theorem \cite{clarke1990optimization}}. 
Note that $\dot{V}$ also satisfies $\dot{V}(x,u)=D^+V(x(t,x_0, u))$
along the trajectories of the system \eqref{sysfxts} for all $t\ge 0$.
It is also useful to note that if $V$ is $\cc^1$, then $\dot{V}$ satisfies $\dot{V}(x,u):=\nabla V(x)^\top f(x,u)$. When $u(t)\equiv 0$ in \eqref{sysfxts}, we simply use $\dot{V}(x):=\dot{V}(x,0)$. Based on this notation, the following definitions are also borrowed from the standard literature on fixed-time stability \cite{6104367,LOPEZRAMIREZ2020104775,fxtiss_converse}:
%where $x\in\re^n$ is the state and $f:\re^n\to\re^n$ is a continuous function that satisfies $f(0)=0$. The following definitions are borrowed from \cite{6104367}:
%

\vspace{0.1cm}
\begin{definition}\label{def_fxts_lf}
    A locally Lipschitz continuous function $V:\re^n\to\re_{\ge 0}$ is said to be a FxTS Lyapunov function for \eqref{sysfxts} if
    \begin{enumerate}
        \item There exists $\underline{\alpha}, \overline{\alpha}\in\ck_\infty$ such that
        \begin{equation}\label{sandw}
            \underline{\alpha}(|x|)\le V(x)\le \overline{\alpha}(|x|),~~~\forall~x\in\mathbb{R}^n.
        \end{equation}
        \item There exists $\Psi\in\ckf$ such that the following holds:
\begin{equation}\label{fxts_lyapunov_bd}
       \dot{V}(x)\le -\Psi(V(x)),~~\forall~x\in\mathbb{R}^n.
    \end{equation}
    \end{enumerate}
\end{definition}
\vspace{0.2cm}

\begin{remark}
In \cite[Lemma 1]{6104367}, it was shown that the origin of \eqref{sysfxts} is FxTS if it admits a FxTS Lyapunov function. Moreover, since $\Psi$ in \eqref{fxts_lyapunov_bd} is given by $\Psi(s)=c_1 s^{p_1}+c_2s^{p_2}$ with $c_1,c_2>0, p_1\in (0,1), p_2>1$, the settling time of $\beta$ satisfies the following upper bound for all $x_0\in \re^n$:
\begin{equation*}
    T(|x_0|)\le \frac{1}{c_1(1-p_1)}+\frac{1}{c_2(p_2-1)}.
\end{equation*}
\end{remark}
% \subsection{Fixed-time Input-to-State Stability}
% Consider the following class of dynamical systems that also depend on an input:
% \begin{equation}\label{sysfxts}
%     \dot{x}=f(x,u),\quad x(0)=x_0,
% \end{equation}
% where $x\in\re^n$ is the state, $u\in\mathcal{L}_\infty^m$ is the input, and $f:\re^n\times\re^m\to\re^n$ is a continuous, non-Lipschitz vector field that satisfies $f(0,0)=0$. Given $x_0\in\re^n$ and $u\in\mathcal{L}_\infty^m$, a solution to \eqref{sysfxts} is represented as $x(t,x_0,u)$ for $t\ge 0$, where $x(0, x_0,u)=x_0$. If $x_0$ and $u$ are clear from the context, we will use $x(t):=x(t, x_0, u)$. For a continuous function $V:\re^n\to\re$, the time derivative of $V$ along the solutions to \eqref{sysfxts} is defined as
% %
% \begin{equation*}
%     \dot{V}(t):=D^+V(x(t))
% \end{equation*}
% %
% for any $t\ge 0$.
% The upper right directional Dini derivative  of $V$ at a point $x\in\re^n$ in the direction $v\in\re^n$ is defined to be
% \begin{equation*}
%     DV(x;v)=\limsup_{h\to 0^+}\frac{V(x+hv)-V(x)}{h}.
% \end{equation*}
% If $V:\re^n\to\re$ is also locally Lipschitz, then its derivative along the solutions to \eqref{sysfxts} also satisfies
% \begin{equation*}
%     \dot{V}(t)=DV(x(t); f(x(t),u(t))).
% \end{equation*}
% It is useful to note that if $V$ is $\cc^1$, then $\dot{V}$ satisfies the following
% \begin{equation*}
%     \dot{V}(t)=\nabla V(x(t))^\top f(x(t),u(t)).
% \end{equation*}
% We will state some definitions and results from \cite{LOPEZRAMIREZ2020104775,fxtiss_converse} that will be particularly useful for our work.

\vspace{0.2cm}
\begin{definition}
    A locally Lipschitz function $V:\re^n\to\re_{\ge 0}$ is said to be a FxT-ISS Lyapunov function for the system \eqref{sysfxts} if it satisfies item 1) of Definition \ref{def_fxts_lf} and 
    there exists $\chi\in\ck_\infty$ and $\Psi\in\ckf$ such that the following holds:
    \begin{equation}\label{fxtiss_imp}
        V(x)\ge \chi(|u|)~~ \implies ~~\dot{V}(x,u)\le -\Psi(V(x)),
    \end{equation}
    for all $x\in\re^n$ and $u\in\re^m$.\QEDB
\end{definition}
\vspace{0.2cm}

It is shown in \cite{LOPEZRAMIREZ2020104775} that if \eqref{sysfxts} admits a FxT-ISS Lyapunov function, then it is FxT-ISS. Converse results can also be found in \cite{fxtiss_converse}.%It  directly follows from definition that if \eqref{sysfxts} is FxT ISS, then the origin of system \eqref{sysfxts} with $u(t)\equiv 0$ is FxTS.
% \begin{remark}
%     It is worth noting that although we allow nonsmooth FxTS Lyapunov functions in Definition \ref{def_fxts_lf} by using the Dini derivative in \eqref{fxtiss_imp} (as done in \cite{6104367}), we require smoothness for FxT ISS Lyapunov functions, since the definition of FxT ISS Lyapunov function presented in \cite{LOPEZRAMIREZ2020104775} requires that $V$ is $\mathcal{C}^1$. The permitted non-smoothness of $V$ in Definition \ref{def_fxts_lf} is important, as the proof of our main result in Section \ref{sec_main} relies on the construction of a non-smooth FxTS Lyapunov function.
% \end{remark}

\section{A LYAPUNOV-BASED FXT-ISS SMALL GAIN THEOREM}\label{sec_main}
\subsection{Main Result}
We consider interconnected systems of the form:
\begin{subequations}\label{sysnu}
    \begin{align}
        \dot{x}_1&=f_1(x_1,x_2, u_1)\label{sysnu1}\\
        \dot{x}_2&=f_2(x_1,x_2, u_2),\label{sysnu2}
    \end{align}
\end{subequations}
where $x_i\in\re^{n_i}$ are the states, $u_i\in \mathcal{L}_\infty^{m_i}$ are the inputs, and $f_i:\re^{n_1}\times\re^{n_2}\times \re^{m_i}\to\re^{n_i}$ are continuous functions with $f_1(0,0,0)=f_2(0,0,0)=0$, see Figure \ref{block}. {Since the vector fields $f_i$ are continuous and the inputs are essentially bounded, the existence of solutions is guaranteed. In this sense, all of our derived stability results are then understood to hold for all maximal solutions}. To study \eqref{sysnu}, we propose a Lyapunov-based small gain approach and view each $x_i$ subsystem as having input $d_i(t)$, where $d_i(t)$ is given by
\begin{equation*}
    d_i(t)=[x_{3-i}(t)\quad u_i(t)]^\top.
\end{equation*}
We now make the following assumption on the system \eqref{sysnu}:

\begin{figure}[t!]
  \centering \includegraphics[width=0.42\textwidth]{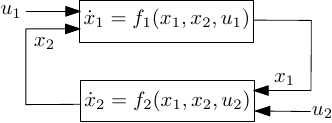}
    \caption{A block diagram depicting the interconnection \eqref{sysnu}.} \label{block}
    %\vspace{-0.4cm}
\end{figure}

\vspace{0.1cm}
\begin{assumption}\label{assump_sysnu}
    Consider the system \eqref{sysnu} and, for each $i=1,2$, there exists locally Lipschitz functions $V_i:\re^{n_i}\to\re_{\ge 0}$ satisfying the following conditions:
    \begin{enumerate}
        \item There exists $\underline{\alpha}_i, \overline{\alpha}_i\in\ck_\infty$ such that
        \begin{equation}\label{assump_sandw}
            \underline{\alpha}_i(|x_i|)\le V_i(x_i)\le \overline{\alpha}_i(|x_i|),
        \end{equation}
        for all $x_i\in\re^{n_i}$.
        \item There exists $\chi_i, \gamma_{i}\in\ck_\infty$ and $\Psi_i\in\ckf$ such that
    \begin{align}\label{fxtiss_imp_assump}
        &V_i(x_i)\ge \max\{\gamma_{i}{(V_{3-i}(x_{3-i}))}, \chi_i(|u_i|)\}\notag\\&
        \implies \dot{V}_i(x_i,d_i)\le -\Psi_i(V_i(x_i)),
    \end{align}
    for all $x_i\in\re^{n_i}$, {$x_{3-i}\in\re^{n_{3-i}}$}, and $u_i\in\re^{m_i}$. \QEDB
    \end{enumerate}
\end{assumption}
\vspace{0.1cm}

In other words, Assumption \ref{assump_sysnu} asks that each $x_i$ subsystem is FxT-ISS with respect to the ``inputs" $x_{3-i}$ and $u_i$, and that there exists a FxT-ISS Lyapunov function that certifies this property. Note that in \eqref{fxtiss_imp_assump}, if $V_i$ is $\cc^1$, then $\dot{V}_i(x_i,d_i)=\nabla V_i(x_i)^\top f_i(x_1,x_2,u_i)$.
Although \eqref{fxtiss_imp_assump} is not exactly in the form \eqref{fxtiss_imp}, it can be placed into the form \eqref{fxtiss_imp} by using \eqref{assump_sandw}. 

\vspace{0.1cm}
We can now state the first main result of the paper:

\vspace{0.1cm}
\begin{thm}[FxT-ISS Small-Gain Theorem]\label{thm_sgt}
    Let system \eqref{sysnu} satisfy Assumption \ref{assump_sysnu} and suppose that, for each $i=1,2$, there exists $\hat{\gamma}_{i}\in\ckp\cup\cki$ such that $\hat{\gamma}_{i}(s)\ge\gamma_{i}(s)$ for all $s>0$ and that the following small-gain condition holds:
    \begin{equation}
    \hat{\gamma}_{1}\circ\hat{\gamma}_{2}(s)<s,~~~\forall~s>0,
    \end{equation}
    %
    %\begin{enumerate}
        %\item $\hat{\gamma}_{i}(s)>\gamma_{i}(s)$ for all $s>0$.
        %\item $\hat{\gamma}_{1}\circ\hat{\gamma}_{2}(s)<s$ for all $s>0$.
    %\end{enumerate}
    %
    Then, the origin of system \eqref{sysnu} is FxT-ISS.\QEDB
\end{thm}

\vspace{0.1cm}
\begin{proof}
    Let $\sigma_\lambda\in\ckp$ take the form $\sigma_\lambda(s)=s^\lambda$ for $\lambda\ge 1$. For each $i,j\in\{1,2\}$, consider the following functions:
    %
    %\begin{subequations}
    \begin{align*}
        \sigma_{i1}(s)&:=\sigma_\lambda(s)\\ \sigma_{i2}(s)&:=\sigma_\lambda\circ\hat{\gamma}_{3-i}(s)\\
        V_{ij}(x_i)&:=\sigma_{ij}\circ V_i(x_i).
    \end{align*}
    where $\lambda$ is a sufficiently large constant such that the pairs of functions $\{V_{11},V_{12}\}$ and $\{V_{21},V_{22}\}$ are both FxT-ISS Lyapunov functions pairs for the subsystems \eqref{sysnu1} and \eqref{sysnu2}, respectively. As shown in Lemmas \ref{lem_k1fxt} and \ref{lem_invfxt} in the Appendix, such a $\lambda\ge 1$ always exists.

    %
   % where $\lambda$ is a sufficiently large constant such that $\sigma_\lambda\circ\hat{\gamma}_{2}\circ V_1(x_1)$ and $\sigma_\lambda\circ\hat{\gamma}_{1}\circ V_2(x_2)$ are FxT-ISS Lyapunov functions of the subsystems \eqref{sysnu1} and \eqref{sysnu2}, respectively. As shown in Lemmas \ref{lem_k1fxt} and \ref{lem_invfxt}, such a $\lambda\ge 1$ always exists. Moreover, since $\lambda\ge 1$, it follows from Lemma \ref{lem_k1fxt} that $\sigma_\lambda\circ V_1(x_1)$ and $\sigma_\lambda\circ V_2(x_2)$ are also FxT-ISS Lyapunov functions of the subsystems \eqref{sysnu1} and \eqref{sysnu2}, respectively. Then, for $i,j\in\{1,2\}$, we define
    %
    %\begin{subequations}
    %\begin{align*}
    %    \sigma_{i1}(s)&:=\sigma_\lambda(s)\\
    %    \sigma_{i2}(s)&:=\sigma_\lambda\circ\hat{\gamma}_{3-i}(s)\\
    %    V_{ij}(x_i)&:=\sigma_{ij}\circ V_i(x_i).
    %\end{align*}
    %\end{subequations}
    %
    Let $x(0)\in\mathbb{R}^n$ and let $x(\cdot)$ be a maximal solution to system \eqref{sysnu} from $x(0)$, defined for all $t\in[0,T_\text{max})$, with $T_\text{max}\in(0,\infty]$. Since each $V_{ij}(x_i)$ is a FxT-ISS Lyapunov function for the $x_i$ subsystem, it satisfies the following properties for $i,j\in \{1,2\}$:
    \begin{enumerate}
        \item The inequalities
        \begin{equation*}
            \underline{\alpha}_{ij}(|x_i|)\le V_{ij}(x_i)\le \overline{\alpha}_{ij}(|x_i|),
        \end{equation*}
        hold for all $x_i\in\re^{n_i}$, where $\underline{\alpha}_{ij}=\sigma_{ij}\circ\underline{\alpha}_i$ and $\overline{\alpha}_{ij}=\sigma_{ij}\circ\overline{\alpha}_i$.
        \item There exists $\Psi_{ij}\in\ckf$ such that:
        \begin{align}\label{subiss}
            &V_{ij}(x_i)\ge \max\{\gamma_{ij}\circ V_{3-i}(x_{3-i}), \chi_{ij}(|u_i|)\}\notag\\&\implies \dot{V}_{ij}(x_i,d_i)\le -\Psi_{ij}(V_{ij}(x_i)),
        \end{align} 
        for all $x_i\in\re^{n_i},$ $x_{3-i}\in\re^{n_{3-i}}$, and $u_i\in\re^{m_i}$, where $\gamma_{ij}:=\sigma_{ij}\circ\hat{\gamma}_i$ and $\chi_{ij}:=\sigma_{ij}\circ \chi_i$.
    \end{enumerate}
    Next, for all $x\in\mathbb{R}^n$, we define
    \begin{align}
        V(x)&:=\max\left\{V_{11}(x_1), V_{22}(x_2), V_{21}(x_2),V_{12}(x_1)\right\},\label{lfmax}
    \end{align}
    and for all $s\ge 0$:
    \begin{align*}
        \underline{\alpha}(s)&=\min_{i,j=1,2}\underline{\alpha}_{ij}\left(\frac{s}{\sqrt{2}}\right),\quad
        \overline{\alpha}(s)=\max_{i,j=1,2}\overline{\alpha}_{ij}(s).
    \end{align*}
    We can verify that $V$ satisfies \eqref{sandw} with $\underline{\alpha}, \overline{\alpha}\in\ck_\infty$. Let $\chi(\cdot):=\max_{i,j=1,2}\chi_{i,j}(\cdot)$, $u(t):=[u_1(t), u_2(t)]^\top$. For each $t\in [0, T_\text{max})$ where $V(x(t))= V_{ij}(x_i(t))$ and $V(x(t))\ge \chi(|u(t)|)$, we have four possible cases (to simplify presentation, we omit the time dependency from $x(t)$):
    \begin{enumerate}
        \item If $i=1$ and $j=1$, we have that $\sigma_{11}\circ V_1(x_1)\ge \sigma_{22}\circ V_2(x_2)$. This implies $V_1(x_1)\ge \hat{\gamma}_1\circ V_2(x_2)$, and thus $V_{11}(x_1)\ge \gamma_{11}\circ V_2(x_2)$. By \eqref{subiss}, we have $\dot{V}_{11}(x_1,d_1)\le -\Psi_{11}(V_{11}(x_1))$.  
        \item If $i=1$ and $j=2$, we have $\sigma_{12}\circ V_1(x_1)\ge \sigma_{21}\circ V_2(x_2)$, which implies $\hat{\gamma}_2\circ V_1(x_1)\ge V_2(x_2)$. By the small gain condition, we also have $V_1(x_1)\ge \hat{\gamma}_1\circ V_2(x_2)$, which implies $V_{12}(x_1)\ge \gamma_{12}\circ V_2(x_2)$. By \eqref{subiss}, we have $\dot{V}_{12}(x_1,d_1)\le -\Psi_{12}(V_{12}(x_1))$. 
        \item If $i=2$ and $j=1$, we have $\sigma_{21}\circ V_2(x_2)\ge \sigma_{12}\circ V_1(x_1)$, which implies $V_2(x_2)\ge \hat{\gamma}_2\circ V_1(x_1)$, and thus $V_{21}(x_2)\ge \gamma_{21}\circ V_1(x_1)$. By \eqref{subiss}, we have $\dot{V}_{21}(x_2,d_2)\le -\Psi_{21}(V_{21}(x_2))$. 
        \item If $i=2$ and $j=2$, we have $\sigma_{22}\circ V_2(x_2)\ge \sigma_{11}\circ V_1(x_1)$, which implies $\hat{\gamma}_1\circ V_2(x_2)\ge V_1(x_1)$. By the small gain condition, we have $V_2(x_2)\ge \hat{\gamma}_2\circ V_1(x_1)$, which implies $V_{22}(x_2)\ge \gamma_{22}\circ V_1(x_1)$. By \eqref{subiss}, we have $\dot{V}_{22}(x_2,d_2)\le -\Psi_{22}(V_{22}(x_2))$. 
    \end{enumerate}
    Since $\Psi_{i,j}\in\ckf$ for each $i,j$, it follows from Lemma \ref{lem_fxtbd} that there exists some $\Psi\in\ckf$ such that 
    \begin{equation*}
        \Psi(s)\le \min_{i,j=1,2}\Psi_{ij}(s),\quad \forall s\ge 0.
    \end{equation*}
    Let $I(t):=\{(i,j) : V(x(t))=V_{ij}(x_i(t))\}$ for $t\in [0, T_\text{max})$,
    and consider the upper-right Dini derivative of $V$ along the trajectories of \eqref{sysnu}, which satisfies
    \begin{align*}
        D^+V(x(t))
        % &=D^+\max_{(i,j)\in I(t)}V_{ij}(x_i(t))\\
        &=\max_{(i,j)\in I(t)} D^+ V_{ij}(x_i(t))\\
        %&= \max_{(i,j)\in I(t)}\dot{V}_{ij}(x_i(t),d_i(t))\\
        &\le \max_{(i,j)\in I(t)}-\Psi_{ij}(V_{ij}(x_i(t)))\\
        &\le \max_{(i,j)\in I(t)} -\Psi(V_{ij}(x_i(t)))\\
        &=-\Psi(V(x(t))),
    \end{align*}
    where the first equality follows from \cite[Lemma 2.9]{giorgi1992dini}. The trajectories are bounded on $[0, T_\text{max})$, and hence they are defined for all $t\ge 0$. This concludes the proof.
\end{proof}
{\begin{remark}An estimate on the settling time for the interconnected system can be obtained from $\Psi$, which depends on a multitude of parameters that appear in the Appendix, some of which are defined implicitly. While this may be computed numerically, it would be of interest for future work to explicitly characterize how the settling time estimate of the interconnected system depends on the settling time estimates of the subsystems.\QEDB
\end{remark}}

Our approach is inspired by the ideas from \cite{jiang1996lyapunov,sgfinite}, in that we leverage the gain functions $\gamma_i$ to construct a FxT-ISS Lyapunov function candidate defined in the max form \eqref{lfmax}. To ensure that the FxT-ISS Lyapunov function properties are preserved, we also utilize the power-function-based scaling technique introduced in \cite{sgfinite}. However, since \cite{sgfinite} studies \emph{finite time} stability notions, it is sufficient in their setting for the gains $\hat{\gamma}_i(\cdot)$ to only have a class $\ckp$ approximation near the origin. On the other hand, since we consider the global notion of \emph{fixed-time} stability, we require that the class $\ckp$ approximation holds globally. One challenge that arises is that if $\hat{\gamma}_1,\hat{\gamma}_2\in\ckp$, the small-gain condition will never be satisfied if there is a $\hat{\gamma}_i$ containing terms with different powers. To address this limitation, we also consider cases where $\hat{\gamma}_i\in\cki$, and we show in Lemma \ref{lem_invfxt} of the Appendix that such functions can also be appropriately scaled to preserve the FxT-ISS Lyapunov function property. {It is also useful to note that if $\gamma_i\in\ckp\cup\cki$, which is a homogeneous and power-type form that commonly arises in Lyapunov-based FxT ISS analysis, one can simply pick $\hat{\gamma}_i(s)=\gamma_i(s)$.}

\vspace{0.1cm}
\begin{remark}
    While most functions in the class $\cki$ do not have a closed-form expression, this is not particularly problematic for the use of Theorem \ref{thm_sgt}. Indeed, if $\hat{\gamma}_i\in\cki$, then $\hat{\gamma}_i^{-1}\in\ckp$, so we can use the following equivalence 
    \begin{equation*}
        V_i(x_i)\ge \hat{\gamma}_i(V_j(x_j))\Longleftrightarrow \hat{\gamma}_i^{-1}(V_i(x_i))\ge V_j(x_j),
    \end{equation*}
    to verify \eqref{fxtiss_imp_assump} in a simplified manner. Moreover, if $\hat{\gamma}_1\in\ckp$ and $\hat{\gamma}_2\in\cki$, we can leverage the following fact
    \begin{equation}\label{sgeq}
        \hat{\gamma}_1\circ\hat{\gamma}_2(s)<s\Longleftrightarrow \hat{\gamma}_1(s)<\hat{\gamma}_2^{-1}(s),
    \end{equation}
    to verify the small-gain condition. Since $\hat{\gamma}_2^{-1}\in\ckp$, the condition $\hat{\gamma}_1(s)<\hat{\gamma}_2^{-1}(s)$ can be checked using straightforward methods, such as Lemma \ref{lem_sandw}. %This is further detailed in the following example.
    \QEDB
\end{remark}

\vspace{-0.3cm}
\subsection{A Second-Order Illustrative Example}
To illustrate the flexibility of Theorem \ref{thm_sgt}, we consider the interconnection between two FxTS systems containing arbitrarily many homogeneous cross terms. In particular, consider the dynamics
\begin{subequations}\label{ex0}
    \begin{align}
        \dot{x}&=f(x)+\sum_{i=1}^{n_1}\varepsilon_{1,i}\sg{y}^{\eta_{1,i}}+u_1\label{ex_x}\\
        \dot{y}&=g(y)+\sum_{j=1}^{n_2}\varepsilon_{2,j} \sg{x}^{\eta_{2,j}}+u_2\label{ex_y},
    \end{align}
\end{subequations}
where $\sg{\cdot}^a=|\cdot|^a\text{sgn}(\cdot)$, with $x,y,u_1,u_2\in\re$, and $\eta_{1,i}>0, \eta_{2,j}>0$ for all $i\in [n_1], j\in [n_2]$. We assume that the origins of the systems $\dot{x}=f(x)$ and $\dot{y}=g(y)$ are FxTS. We also assume that they admit FxTS Lyapunov functions that satisfy the following assumption
\begin{assumption}\label{assump_ex0}
    The functions $f:\re\to\re$ and $g:\re\to\re$ are continuous and there exist $\cc^1$ functions $V:\re\to\re_{\ge 0}$ and $W:\re\to\re_{\ge 0}$ satisfying the following properties:
    \begin{enumerate}
        \item There exists $\un{c}_1,\ov{c}_1,\un{c}_2, \ov{c}_2>0$ and $\kappa_1, \kappa_2>0$ such that
        \begin{equation*}
            \un{c}_1 |x|^2\le V(x)\le \ov{c}_1 |x|^2,\quad  \un{c}_2 |y|^2\le W(y)\le \ov{c}_2 |y|^2
        \end{equation*}
        and $|\nabla V(x)|\le \kappa_1 |x|,\ |\nabla W(y)|\le \kappa_2|y|$ for all $x,y\in\re$.
        \item There exists $a_1, a_2>0, p_1, p_2\in (0,1)$, and $q_1, q_2>1$ such that:
    \begin{subequations}
    \begin{align*}
        \nabla V(x)f(x)&\le -a_1 V^{p_1}(x)-a_1 V^{q_1}(x)\\ \nabla W(y)g(y)&\le -a_2 W^{p_2}(y)-a_2 W^{q_2}(y),
    \end{align*}
    \end{subequations}
    for all $x,y\in\re$.\QEDB
    \end{enumerate}
\end{assumption}
We will show that, under the natural conditions of Assumption \ref{assump_ex0}, system \eqref{ex0} is FxT-ISS, as long as $\varepsilon_1, \varepsilon_2$ are sufficiently small in magnitude and the exponents $\eta_{1,i}, \eta_{2,j}$ satisfy some conditions specified below.
\begin{prop}\label{thm_ex}
    For the system \eqref{ex0}, suppose that Assumption \ref{assump_ex0} holds, where $\eta_{1,i}, \eta_{2,j}>0$ and $\eta_{1,i} \eta_{2,j}\in [p_2, q_2]$ for each $i\in [n_1]$, and $j\in [n_2]$, and
    \begin{equation}\label{conditionthm2}
        \frac{\max_j \eta_{2,j}}{\min_k \eta_{2,k}}\le \frac{q_2}{p_2}.
    \end{equation}
    Denote $\varepsilon_k=\max_i |\varepsilon_{k,i}|$ for $k=1,2$. Then, there exists $\varepsilon>0$ such that if
    \begin{equation}\label{epbd}
        |\varepsilon_2|\max_{j}|\varepsilon_1|^{\eta_{2,j}}<\varepsilon,
    \end{equation}
     then, the origin is FxT-ISS for system \eqref{ex0}.\QEDB
\end{prop}
\vspace{0.1cm}
\begin{proof}
   Consider the FxT-ISS Lyapunov functions $V(x)$ and $W(y)$ from Assumption \ref{assump_ex0}, respectively. Since we can always scale $V$ and $W$ while still satisfying Assumption \ref{assump_ex0}, we will assume, without loss of generality, that $\un{c}_1=\un{c}_2=1$. Along the trajectories of the $x$ subsystem, we have
    \begin{align*}
 \dot{V}(x)
 % &= \nabla V(x)f(x)+\nabla V(x)\left(\varepsilon_1 \sum_{i=1}^{n_1}\sg{y}^{\xi_i}+u_1\right)\\
 %        &\le -a_1 V^{p_1}(x)-a_1 V^{q_1}(x)+\frac{\delta}{2} |\nabla V(x)|^2\\
 %        &~~~~+\frac{1}{2\delta} \left(\varepsilon_1\sum_{i=1}^{n_1} \sg{y}^{\xi_i}+u_1\right)^2\\
        &\le -a_1 V^{p_1}(x)-a_1 V^{q_1}(x)+\frac{a_1}{4} V(x)\\
        &~~~~+\frac{\kappa_1^2}{a_1} \left(\sum_{i=1}^{n_1} \varepsilon_{1,i}\sg{y}^{\eta_{1,i}}+u_1\right)^2\\
        &\le-\frac{3a_1}{4} V^{p_1}(x)-\frac{3a_1}{4} V^{q_1}(x)\\
        &~~~~+\frac{(n_1+1)\kappa_1^2\varepsilon_1^2}{ a_1}\sum_{i=1}^{n_1} W^{\eta_{1,i}}(y)+\frac{(n_1+1)\kappa_1^2}{a_1} u_1^2,
    \end{align*}
where the first inequality follows by Young's inequality with item 1) in Assumption \ref{assump_ex0}, and the second inequality follows by applying Lemma \ref{jensenlemma}, where we again leverage item 1) in Assumption \ref{assump_ex0}.

We can perform the same computations for the $y$-subsystem to obtain the following inequality:
\begin{align*}
    \dot{W}(y)&\le -\frac{3a_2}{4} W^{p_2}(y)-\frac{3a_2}{4} W^{q_2}(y)\\&~~~+\frac{(n_2+1)\kappa_2^2\varepsilon_2^2}{ a_2}\sum_{i=1}^{n_2} V^{\eta_{2,i}}(x)+\frac{(n_2+1)\kappa_2^2}{ a_2} u_2^2.
\end{align*}
Next, for $i=1,2$, we define the following functions:
\begin{subequations}
    \begin{align}
        \alpha_i(s)&:=\frac{a_i}{4}(s^{p_i}+s^{q_i})\\
        \sigma_i(s)&:=\frac{(n_i+1)\kappa_i^2\varepsilon_i^2}{ a_i}\sum_{j=1}^{n_i} s^{\eta_{i,j}}\\
        %\sigma_2(s)&:=\frac{(n_2+1)\kappa_2^2\varepsilon_2^2}{ a_2}\sum_{j=1}^{n_2} s^{\eta_{2,j}}\\
        \hat{\chi}_i(s)&:=\frac{(n_i+1)\kappa_i^2}{ a_i}s^2.
    \end{align}
\end{subequations}
Let $\gamma_i:=\alpha_i^{-1}\circ \sigma_i$ and $\chi_i:=\alpha_i^{-1}\circ\hat{\chi}_i$ for $i=1,2$. Then, we obtain:
\begin{align*}
    V(x)&\ge \max\{\gamma_1(W(y)), \chi_1(|u_1|) \}\Rightarrow \dot{V}(x)\le -\alpha_1(V(x))\\
    W(y)&\ge \max\{\gamma_2(V(x)), \chi_2(|u_2|) \}\Rightarrow \dot{W}(y)\le -\alpha_2(W(y)).
\end{align*}
By Lemma \ref{lem_diss} in the Appendix, there exists $C>0$ such that the inequality
\begin{equation}\label{ex0_cond}
    \frac{(n_2+1)\kappa_2^2\varepsilon_2^2}{a_2}\max_j \left(\frac{(n_1+1)\kappa_1^2\varepsilon_1^2}{a_1}\right)^{\eta_{2,j}}< C,
\end{equation}
implies the existence of $\hat{\gamma}_1, \hat{\gamma}_2\in\ckp\cup\cki$ such that $\hat{\gamma}_i(s)\ge\gamma_i(s)$ and $\hat{\gamma}_i\circ\hat{\gamma}_{3-i}(s)<s$ for all $s>0$ and $i=1,2$. Inequality \eqref{ex0_cond} holds if $\varepsilon_1,\varepsilon_2$ satisfy the following condition:

\vspace{-0.5cm}
\begin{small}
\begin{equation*}
    |\varepsilon_2|\max_j|\varepsilon_1|^{\eta_{2,j}}<\sqrt{\frac{Ca_2}{(n_2+1)\kappa_2^2}\min_j \left(\frac{a_1}{(n_1+1)\kappa_1^2}\right)^{\eta_{2,j}}}:=\varepsilon.
\end{equation*}
\end{small}

\noindent 
Therefore, we can then apply Theorem \ref{thm_sgt} to obtain the result.
\end{proof}
\vspace{0.1cm}

To verify our results numerically, we plot the trajectories of system \eqref{ex0} with varying initial conditions, where the parameters are chosen to yield the following system
\begin{subequations}\label{exdplot}
    \begin{align}        \dot{x}&=-x^\frac13-x^3+0.1\left(\sg{y}^{2}+\sg{y}^{2.5}\right){+u_1(t)}\\
        \dot{y}&=-\sg{y}^\frac12-\sg{y}^\frac94-0.3\sg{x}^{\frac49}+0.2\sg{x}^{\frac58}{+u_2(t)},
    \end{align}
\end{subequations}
{with $u_1(t)=\sin(1.7t)+\frac12\sin(5.1t)$ and $u_2(t)=\frac45\cos(2.3t)+0.35\sin(6.7t)$.} The trajectories are shown in Figure \ref{fxt1plot}, illustrating the FxT ISS property.
\begin{figure}[t!]
  \centering \includegraphics[width=0.44\textwidth]{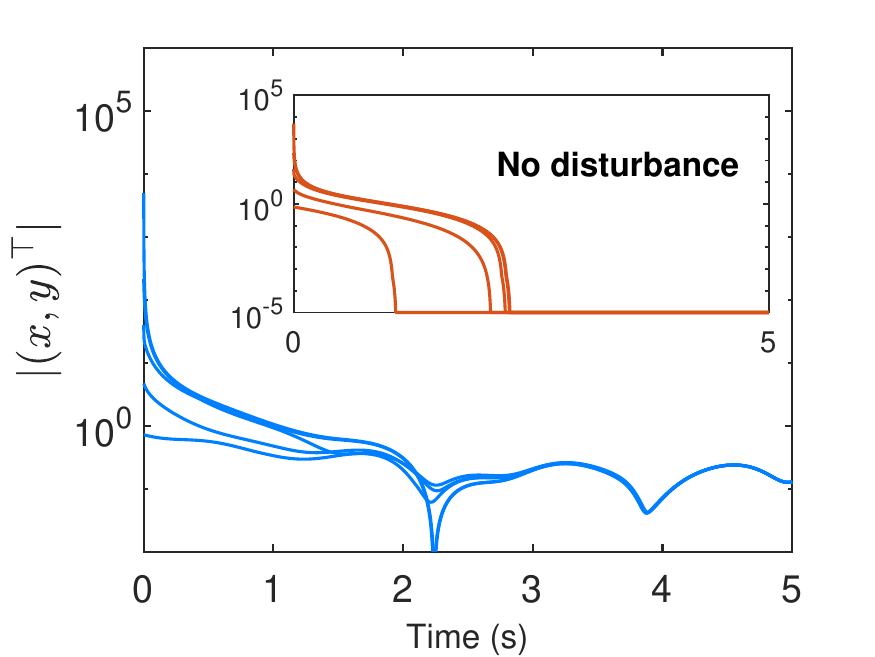}
    \caption{{Trajectories of system \eqref{exdplot} with varying initializations.}} \label{fxt1plot}
    \vspace{-0.4cm}
\end{figure}

\begin{remark}
The previous example illustrates the applicability of Theorem~\ref{thm_sgt} to settings involving nonlinearities in~$\gamma_i$, which are not directly addressed by the small-gain conditions presented in~\cite{10886634}. For instance, by the structure of $\gamma_1(W(y))$, we notice that $\gamma_1(W(y))>|y|$ either when $y$ is small or large in magnitude, rendering the results of \cite{10886634} incompatible for our particular setting (see also Remark 2 and 4 of \cite{10886634}). 
An important observation is that the presence of multiple exponents of $V$ in the FxT-ISS Lyapunov bounds in \eqref{fxtiss_imp} allows our results to handle $\ckp$ gains and interconnections that contain terms with different exponents. For standard Lyapunov-based nonlinear small-gain theory \cite{jiang1996lyapunov}, the structure of the rate of decrease of the Lyapunov function (upon satisfaction of the ISS condition) results in a limited class of admissible gain functions. Even in the case of finite-time ISS \cite{sgfinite} for interconnections, gain functions in the class $\ckp$ often only contain one term, and the exponents of the different gain functions need to multiply to 1. On the other hand, Theorem \ref{thm_sgt} relaxes this condition by leveraging the different exponents of $V$ and by using Lemma \ref{lem_sandw}. In particular, instead of requiring $\eta_{1,i}\eta_{2,j}=1$ for each $i\in[n_1], j\in[n_2]$, we now only require $\eta_{1,i}\eta_{2,j}\in [p_2, q_2]$, which is a neighborhood of 1.\QEDB
\end{remark}
\vspace{0.1cm}
\begin{remark}
The condition \eqref{epbd} essentially represents a trade-off between the weights of the disturbance terms in each subsystem, which is quite standard in small-gain theory. Moreover, we also allow for arbitrarily many cross terms, as long as the exponents of the cross terms for one of the subsystems are grouped sufficiently ``close" together (this is related to the condition $\frac{\max_j \eta_{2,j}}{\min_k \eta_{2,k}}\le \frac{q_2}{p_2}$ in Proposition \ref{thm_ex}). It is also useful to note that the system \eqref{ex0} is symmetric in $x$ and $y$, so if the conditions $\frac{\max_j \eta_{2,j}}{\min_k \eta_{2,k}}\le \frac{q_2}{p_2}$ and $\eta_{1,i}\eta_{2,j}\in [p_2, q_2]$ do not hold, then we can switch the role of $x$ and $y$ and instead check a new set of conditions. \QEDB
\end{remark}
\vspace{0.1cm}
% \subsection{Connections with FxT ISS in dissipative form}
% The previous example reveals a potential extension of our results to account for systems that admit FxT ISS Lyapunov functions in dissipative form, which is defined to be the following
% \begin{definition}
%     A continuous function $V:\re^n\to\re_{\ge 0}$ is said to be a dissipative FxT ISS Lyapunov function for \eqref{sysfxts} if it satisfies item 1) of Definition \ref{def_fxts_lf} and 
%     there exists $\sigma\in\ck_\infty$ and $\Psi\in\ckf$ such that the following holds
%     \begin{equation}\label{fxtiss_diss}
%         D^+V(x(t))\le -\Psi(V(x(t)))+\sigma(|u(t)|),
%     \end{equation}
%     along the trajectories of \eqref{sysfxts} for all $t\ge 0$.\QEDB
% \end{definition}

% It is quite straightforward to verify that \eqref{sysfxts} admits a dissipative FxT ISS Lyapunov function if and only if it admits a FxT ISS Lyapunov function.

\section{FIXED-TIME FEEDBACK OPTIMIZATION WITHOUT TIMESCALE SEPARATION}\label{sec_fbkopt}
In this section, we use the FxT-ISS Small-Gain Theorem presented in Section \ref{sec_main} to study feedback optimization problems in dynamical systems with time-varying cost functions, \emph{without using time scale separation}. Specifically, traditional feedback optimization techniques require inducing a timescale separation between the plant and the controller, as it enables the system designer to leverage tools from singular perturbation theory for the stability analysis. {Such techniques are widely used for the real-time optimization of real-world infrastructure, including power systems, traffic networks, and cyber-physical systems \cite{9075378,8673636,bianchin2022online,9540998}.} 
In the context of FxT stability, this technique was explored recently in \cite{11239427}, where the authors used a composite Lyapunov function to establish fixed-time convergence for certain classes of multi-timescale feedback optimization schemes. However, as acknowledged by the authors, the derived timescale separation requirement can be highly conservative. {Related ideas were also recently explored in \cite{liu2023singular}, where Lyapunov-based small-gain arguments were used within a generalized singular perturbation framework with state-dependent perturbation functions, including an application to feedback optimization. In contrast, the approach we take here is not based on selecting perturbation functions or enforcing a singular perturbation structure. Instead, we use the proposed FxT-ISS small-gain theorem to show that, under suitable cost function structure, nonsmooth fixed-time gradient feedback achieves FxT-ISS tracking of time-varying optimizers \emph{without requiring a timescale separation} between the plant and controller.} 
\subsection{Model and Problem Statement}
We consider general nonlinear plants of the form
\begin{equation}\label{exsys}
    \dot{x}=f(x,u),
\end{equation}
where $x\in\re^n$ is the state, $u\in\mathcal{L}_\infty^N$ is a measurable and essentially bounded control input, and $f:\re^n\times\re^N\to\re^n$ is a continuous function that satisfies the following assumption:

\vspace{0.1cm}
\begin{assumption}\label{assump_ex}
    There exists a $\cc^1$ $\ell$-Lipschitz function $h:\re^N\to\re^n$ and $\gamma>0$ such that:
    \begin{enumerate}[(a)]
    \item $f(h(u),u)=0$ for all $u\in\mathbb{R}^N$.
    \item There exists $a>0, p\in (0,1)$, and $q>1$ such that the following implication holds:
    \begin{align}\label{ex_fxtiss}
        &|x-h(\hat{u})|\ge\gamma |u-\hat{u}|\implies\\&(x-h(\hat{u}))^\top f(x,u)\le -\frac{a}{2}|x-h(\hat{u})|^{2p}-\frac{a}{2}|x-h(\hat{u})|^{2q},\notag
    \end{align}
    for all $x\in\re^n$, $u\in\re^N$, and $\hat{u}\in\re^N$.\QEDB
    \end{enumerate}
\end{assumption}
\vspace{0.1cm}
\begin{remark}
    In the literature, the mapping $h(u)$ is traditionally referred to as the \emph{quasi-steady state mapping} for the dynamics \eqref{exsys}. It is standard in the feedback optimization literature to assume that this mapping exhibits appropriate stability properties uniformly in $u$, see, e.g \cite{11239427, 9075378}. We impose a similar condition, but in the context of FxT-ISS. Essentially, item (b) of Assumption \ref{assump_ex} asks that for each fixed reference input $\hat{u}$, the deviation $x-h(\hat{u})$ is FxT-ISS with respect to the deviation on the ``input" $u-\hat{u}$, which is established via the quadratic-like FxT ISS Lyapunov function $|x-h(\hat{u})|^2$. We require that this property holds uniformly in $x, u$, and $\hat{u}$. {For
suitable controllable or feedback-linearizable plants, this type of
FxT-ISS tracking property can be enforced via standard methods such as homogeneous
feedback \cite{10018220}, terminal/sliding-mode control \cite{zuo2015non}, or implicit Lyapunov
techniques \cite{LOPEZRAMIREZ2020104775}.} \QEDB
\end{remark}
\vspace{0.1cm}
\begin{example}
    One class of systems satisfying Assumption \ref{assump_ex} are those of the following form
    \begin{align}\label{accplant}
    \dot{x}=-\frac{A_1(x-h(u))}{|x-h(u)|^{\tilde{p}}}- \frac{A_2(x-h(u))}{|x-h(u)|^{\tilde{q}}},
\end{align}
where $A_1, A_2\in\re^{n\times n}$ with $A_1, A_2\succ 0$, $\tilde{p}\in (0,1)$, and $\tilde{q}<0$. The system \eqref{accplant} is a significant generalization of the plants considered in \cite{10644358}, where the authors used $N=n$, $A_1=A_2=I_n$, and $h(u)=u$. 

For more insight into system \eqref{accplant}, consider the scalar system $\dot{x}=wx+cv$
where $w,c\in\re\setminus\{0\}$ are known parameters and $v$ is a control input. Then, the following state feedback law:
\begin{equation*}
    v=-\frac{w}{c}h(u)-\frac{2|w|}{c}\left(\frac{x-h(u)}{|x-h(u)|^{\tilde{p}}}+ \frac{x-h(u)}{|x-h(u)|^{\tilde{q}}}\right),
\end{equation*}
where $h$ is $\ell$-Lipschitz and $u\in\re$ is an external input, renders the system into a closed form that satisfies Assumption \ref{assump_ex}. Extensions to multi-variable settings can be done using results from the FxT control literature \cite{6104367}.
\QEDB 
\end{example}

\vspace{0.1cm}
The primary goal is to design an update law on the input $u$ that stabilizes \eqref{exsys}, in a fixed-time, while ensuring that its trajectories converge to the solutions of following the time-varying optimization problem
\begin{subequations}\label{originalfbkproblem}
\begin{align}
    &\min_{x,u}\phi_\theta(x,u)\label{fbk_con}\\
    &\text{subject to: } x=h(u)\label{fbk_const},
\end{align}
\end{subequations}
where the cost function $\phi_\theta:\re^n\times\re^N\to\re$ is $\cc^1$ and depends on a time-varying parameter $\theta\in\re^m$, that evolves according to the following exosystem:
\begin{figure}[t!]
  \centering \includegraphics[width=0.3\textwidth]{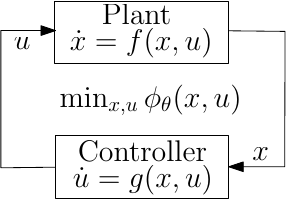}
    \caption{A block diagram depicting the feedback optimization scheme \eqref{ex_fbk_sys}.} \label{fxtfbk}
    %\vspace{-0.4cm}
\end{figure}
\begin{equation}\label{dtheta}    \dot{\theta}=\varepsilon_0\Pi(\theta),\quad\theta\in\Theta,
\end{equation}
where $\varepsilon_0>0$ is a small parameter that controls the rate of change of $\theta$. For the purpose of regularity, the set $\Theta\subset\re^m$ and mapping $\Pi:\re^m\to\re^m$ satisfy the following assumption:
\begin{assumption}\label{assumpt}
    The function $\Pi(\cdot)$ is Lipschitz continuous, and the set $\Theta$ is compact and forward invariant under the dynamics \eqref{dtheta}.\QEDB
\end{assumption}
\vspace{0.1cm}

By substituting \eqref{fbk_const} into \eqref{fbk_con}, we arrive at the following unconstrained quasi-steady state optimization problem
\begin{equation}\label{fbk_uncon}
    \min_u \Phi_\theta(u),
\end{equation}
where $\Phi_\theta(u):=\phi_\theta(h(u),u)$. To guarantee that \eqref{fbk_uncon} is well-defined and has a unique solution for each $\theta\in\Theta$, we impose the following standard assumptions on the cost functions $\Phi_\theta(\cdot)$ \cite{9540998, 10189107}:

\begin{assumption}\label{assump_phi}
    {The function $u\mapsto\Phi_\theta(u)$ is $\cc^1$, $L$-smooth and $\mu$-strongly convex, uniformly in $\theta$, i.e., there exists $L>0$ and $\mu>0$ such that $|\nabla\Phi_\theta(\hat{u})-\nabla\Phi_\theta(u)|\le L|\hat{u}-u|$
    and $\Phi_\theta(\hat{u})\ge \Phi_\theta(u)+\nabla \Phi_\theta(u)^\top (\hat{u}-u)+\frac{\mu}{2}|\hat{u}-u|^2$
    hold for all $u,\hat{u}\in\re^N$. Moreover, there exists a $\mathcal{C}^1$ function $\varphi:{\Theta}\to\mathbb{R}^N$ such that $\varphi(\theta)=\argmin_{u} \Phi_\theta(u)$.}\QEDB
\end{assumption}
\vspace{0.1cm}

When $\varepsilon_0=0$ in \eqref{dtheta}, the parameter $\theta$ remains constant, which results in a constant solution $\vpt$ to \eqref{fbk_uncon}. However, when $\varepsilon_0$ is large, then the solutions to \eqref{fbk_uncon} may exhibit fast time-variations and be difficult to track without access to an internal model. Therefore, to address \eqref{fbk_uncon}, we propose a fixed-time gradient-based feedback scheme and study FxT-ISS of the interconnection with respect to the ``input" $\varepsilon_0\Pi(\theta)$.
\subsection{Fixed-Time Gradient-Based Feedback}
To design the controller for system \eqref{exsys}, let $\xi_1\in (0,1)$, $\xi_2<0$, $\xi=(\xi_1,\xi_2)$, and consider the function $\fx:\mathbb{R}^N\to\mathbb{R}^N$ given by:
\begin{equation}\label{scale_function}
    \fx(s)=\frac{s}{|s|^{\xi_1}}+\frac{s}{|s|^{\xi_2}},
\end{equation}
which is continuous at $s=0$ {with $\fx(0)=0$}. To solve the quasi-steady state optimization problem \eqref{fbk_uncon} in fixed-time, we can consider a fixed-time gradient flow on $u$:
\begin{equation*}
    \dot{u}=-\fx(\cg_\theta(h(u),u))
\end{equation*}
where
\begin{equation}
    \cg_\theta(x,u)=H(u)^\top \nabla\phi_\theta(x,u), \ \ H(u)^\top=[\bj_{h}(u)^\top\quad\mathbb{I}_N],\label{ex0ghu}
\end{equation}
{and $\mathbb{I}_N$ denotes the $N$-dimensional identity matrix.}
Note that, using the chain rule, we have $\nabla\Phi_\theta(u)=\cg_\theta(h(u),u)$. Therefore, to obtain a real-time feedback controller to solve problem \eqref{originalfbkproblem}, we replace the steady state approximation, $h(u)$, with its measured value, $x$, leading to the following interconnection:
\begin{subequations}\label{ex_fbk_sys}
\begin{align}
    \dot{x}&=f(x,u)\\
    \dot{u}&=-\fx(\cg_\theta(x,u))\label{ex_fbk_u}\\
    \dot{\theta}&=\varepsilon_0\Pi(\theta).
\end{align}
\end{subequations}
Note that in \eqref{ex_fbk_sys} we do not impose any timescale separation between the dynamics of $x$ and the dynamics of $u$. To study this closed-loop system, we impose the following mild Lipschitz assumption on $\cg_\theta$, which is also standard in the literature \cite{9075378}:
\begin{assumption}\label{assump_fbk_wlip}
    There exists $K>0$ such that
    \begin{equation*}
        |\cg_\theta(\hat{x},u)-\cg_\theta(x,u)|\le K|\hat{x}-x|
    \end{equation*}
    for all $x,\hat{x}\in\re^n$, $u\in\re^N$, and $\theta\in\Theta$.\QEDB
\end{assumption}
\vspace{0.1cm}

The following theorem characterizes the FxT-ISS properties of system \eqref{ex_fbk_sys}:

\vspace{0.1cm}
\begin{thm}\label{thm_fbk}
    For system \eqref{ex_fbk_sys}, with $\xi_1\in(0,1)$ and $\xi_2<0$, suppose that Assumptions \ref{assump_ex}-\ref{assump_fbk_wlip} hold and that:
    \begin{equation}\label{ex1sgt}
        \mu>K(\ell+\gamma).
    \end{equation}
    Then, there exists $\beta\in\klfx$ and $\varrho\in\ck$ such that, for each $x(0)\in\re^n$, $u(0)\in\re^N$, and $\theta(0)\in{\Theta}$, each solution to \eqref{ex_fbk_sys} satisfies
    \begin{align*}
        &\left\lvert\begin{bmatrix}
            x(t)-h(\varphi(\theta(t)))\\ u(t)-\varphi(\theta(t))
        \end{bmatrix}\right\rvert\\&~~~~~~~\le \beta\left(\left\lvert\begin{bmatrix}
            x(0)-h(\varphi(\theta(0)))\\ u(0)-\varphi(\theta(0))
        \end{bmatrix}\right\rvert,t\right)+\varrho(|\varepsilon_0\Pi(\theta(t))|_\infty)
    \end{align*}
    for all $t\ge 0$.\QEDB
\end{thm}
\begin{proof}
   Using the change of coordinates $\tx=x-h(\varphi(\theta))$, $\tu=u-\varphi(\theta)$, we arrive at the following system:
   \begin{subequations}\label{ex1_transformed}
    \begin{align}    \dot{\tx}&=f(\tx+h(\varphi(\theta)), \tu+\varphi(\theta))-\varepsilon_0\bj_h(\varphi(\theta))\bj_\varphi(\theta)\Pi(\theta)\\
    \dot{\tu}&=-\fx(\cg_\theta(\tx+h(\varphi(\theta)), \tu+\varphi(\theta)))-\varepsilon_0\bj_\varphi(\theta)\Pi(\theta).
    \end{align}
    \end{subequations}
    For the $\tx$- and $\tu$-subsystems, consider the $\mathcal{C}^1$ FxT-ISS Lyapunov functions $V(\tx)=|\tx|^2$ and $W(\tu)=|\tu|^2$, respectively. Let $M:=\sup_{\theta\in\Theta}|\bj_\varphi(\theta)|$, which is finite since $\Theta$ is compact. Then, the derivative of $V$ along the trajectories of the $\tx$-subsystem satisfies
    \begin{align*}
        \dot{V}
        &\le 2\tx^\top f(\tx+h(\varphi(\theta)), \tu+\varphi(\theta))+2M\ell|\varepsilon_0\Pi(\theta)||\tx|.
    \end{align*}
    By Assumption \ref{assump_ex}, for $V(\tx)\ge \gamma^2 W(\tu)$, we have
    \begin{align}
        \dot{V}\le -a|\tx|^{2p}-a|\tx|^{2q}+2M\ell|\varepsilon_0\Pi(\theta)||\tx|.\label{ex1_xgu}
    \end{align}
    From \eqref{ex1_xgu}, if we also have $V(\tx)\ge \frac{16M^2\ell^2}{a^2}|\varepsilon_0\Pi(\theta)|^2$, then
    \begin{align}\label{ex1_xsysbd}
        \dot{V}\le -\frac{a}{2}V^p(\tx)-\frac{a}{2}V^q(\tx).
    \end{align}
    Hence, if $V(\tx)\ge\max\{\gamma^2 W(\tu), \frac{16M^2\ell^2}{a^2}|\varepsilon_0\Pi(\theta)|^2\}$, then \eqref{ex1_xsysbd} holds. Next, consider the $\tu$-subsystem, and note that since
    \begin{align}       &\cg_\theta(\tx+h(\varphi(\theta)), \tu+\varphi(\theta))\notag\\&= \cg_\theta(h(\tu+\varphi(\theta)), \tu+\varphi(\theta))\notag\\&~~+\cg_\theta(\tx+h(\varphi(\theta)), \tu+\varphi(\theta))-\cg_\theta(h(\tu+\varphi(\theta)), \tu+\varphi(\theta))\label{gpm},
    \end{align}
    it follows via Assumptions \ref{assump_ex} and \ref{assump_fbk_wlip} that
    \begin{align}
    |\cg_\theta(\tx+h(\varphi(\theta)), \tu+\varphi(\theta))|&\le |\cg_\theta(h(\tu+\varphi(\theta)), \tu+\varphi(\theta))|\notag\\&~~~+ K(|\tx|+\ell |\tu|).\label{ex1_gupbd}
\end{align}  
\begin{figure*}[t!]
  \centering
    \includegraphics[width=0.35\textwidth]{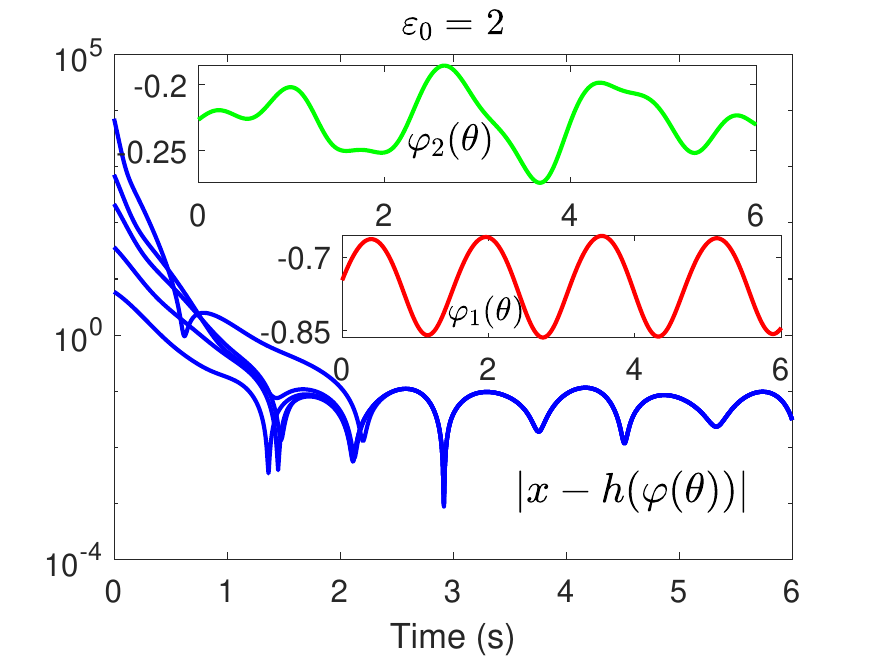}\hspace{-0.5cm}\includegraphics[width=0.35\textwidth]{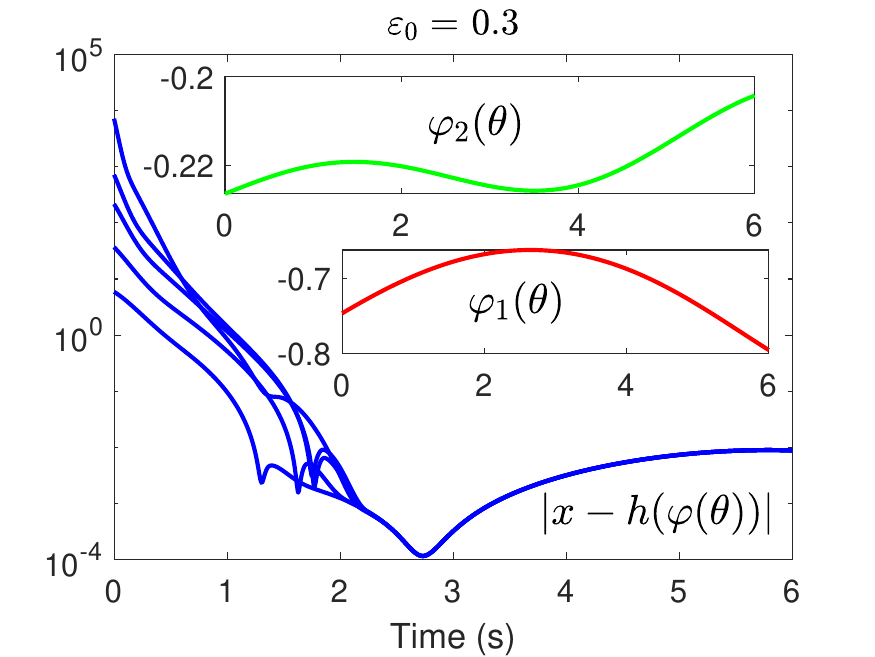}\hspace{-0.5cm}\includegraphics[width=0.35\textwidth]{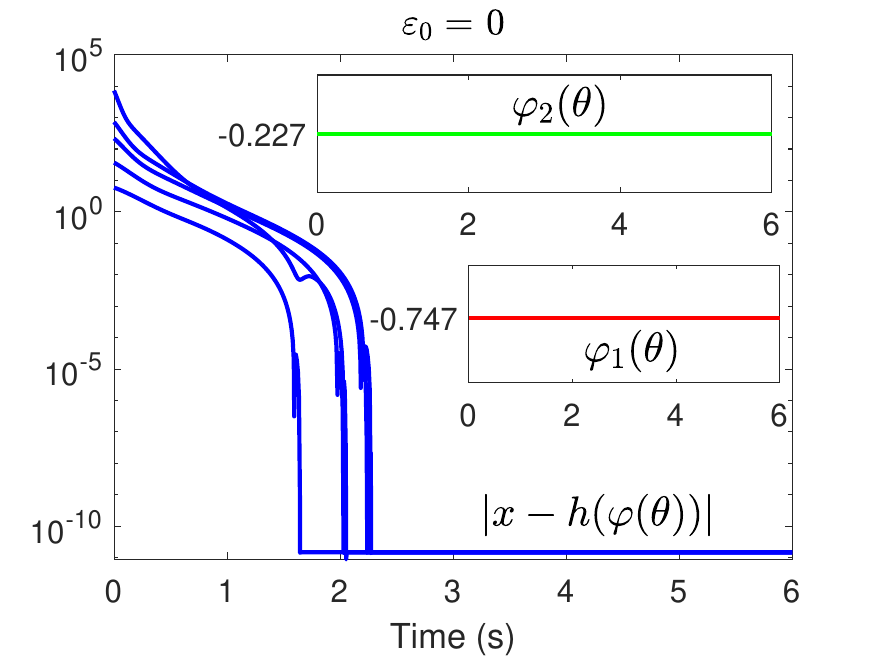}
    \caption{Trajectories of the transformed variable $\tx$ and the time-varying optimizer $\varphi(\theta)$ from the interconnection \eqref{ex1_transformed}, with $\varepsilon_0=2, 0.3, 0$ and varying initial conditions.}\label{ex0plot}  
    \vspace{-0.2cm}
\end{figure*}
Moreover, by Assumption \ref{assump_phi}, we have 
\begin{equation}\label{ex1_gsandw}
    \mu|\tu|\le |\cg_\theta(h(\tu+\varphi(\theta)), \tu+\varphi(\theta))|\le L|\tu|.
\end{equation}
    From \eqref{ex1sgt}, we know that $\mu-\ell K>0$, which implies that we can pick $\delta\in (0, \mu-\ell K)$ sufficiently small, such that $\gamma\gamma_\delta\in (0,1)$, where
    \begin{equation*}
        \gamma_\delta:=\frac{K}{\mu-\ell K-\delta}.
    \end{equation*}
    Suppose that $|\tu|\ge \gamma_\delta|\tx|$. Then, using  \eqref{ex1_gsandw}, we can upper-bound \eqref{ex1_gupbd} as follows:
    \begin{align}\label{upperboundgtheta}
        |\cg_\theta(\tx+h(\varphi(\theta)), \tu+\varphi(\theta))|&\le\left(L+K\left(\frac{1}{\gamma_\delta} +\ell\right)\right)|\tu|\notag\\
        &=(L+\mu-\delta)|\tu|.
    \end{align}
    Similarly, using again \eqref{gpm}, Assumptions \ref{assump_ex} and \ref{assump_fbk_wlip}, and \eqref{ex1_gsandw}, we have:
    \begin{align}
        &|\cg_\theta(\tx+h(\varphi(\theta)), \tu+\varphi(\theta))|\notag\\&\ge |\cg_\theta(h(\tu+\varphi(\theta)), \tu+\varphi(\theta))|- K(|\tx|+\ell |\tu|)\notag\\&\ge \left(\mu-K\left(\frac{1}{\gamma_\delta} +\ell\right)\right)|\tu|\notag\\&=\delta |\tu|,\label{ex1_glowbd}
    \end{align}
    for all $|\tu|\ge \gamma_\delta|\tx|$, and we also obtain:
    \begin{align}
        \tu^\top \cg_\theta(\tx+h(\varphi(\theta)), \tu+\varphi(\theta))&\ge \mu |\tu|^2-K|\tu|(|\tx|+\ell|\tu|)\notag\\
        &\ge \delta|\tu|^2.\label{ex1_ugbd}
    \end{align}
    By combining \eqref{upperboundgtheta} and \eqref{ex1_ugbd}, we obtain
    \begin{equation*}
        \frac{\tu^\top \cg_\theta(\tx+h(\varphi(\theta)), \tu+\varphi(\theta))}{|\cg_\theta(\tx+h(\varphi(\theta)), \tu+\varphi(\theta))|^{\xi_1}}\ge \frac{\delta}{(L+\mu-\delta)^{\xi_1}}|\tu|^{2-\xi_1}.
    \end{equation*}
    Similarly, combining \eqref{ex1_glowbd} and \eqref{ex1_ugbd}, and using the fact that $\xi_2<0$, yields
    \begin{equation*}
        \frac{\tu^\top \cg_\theta(\tx+h(\varphi(\theta)), \tu+\varphi(\theta))}{|\cg_\theta(\tx+h(\varphi(\theta)), \tu+\varphi(\theta))|^{\xi_2}}\ge \delta^{1-\xi_2}|\tu|^{2-\xi_2}.
    \end{equation*}
    Then, we can consider the FxT-ISS Lyapunov function $W(\tu)=|\tu|^2$, to establish that if $W(\tu)\ge \gamma_\delta^2 V(\tx)$, then the derivative of $W$ along the trajectories of the $\tu$-subsystem satisfies:
    \begin{align*}
        \dot{W}&\le -b|\tu|^{2-\xi_1}-b|\tu|^{2-\xi_2}+2M|\varepsilon_0\Pi(\theta)||\tu|.
    \end{align*}
    where
    \begin{equation*}
        b:=\min\left\{\frac{2\delta}{(L+\mu-\delta)^{\xi_1}},\quad 2\delta^{1-\xi_2}\right\}.
    \end{equation*}
    If, additionally, we have  $W(\tu)\ge \frac{16M^2}{b^2}|\varepsilon_0\Pi(\theta)|^2$, then
    \begin{equation}\label{ex1_usysbd}
        \dot{W}\le -\frac{b}{2}W^{1-\frac12\xi_1}(\tu)-\frac{b}{2}W^{1-\frac12\xi_2}(\tu).
    \end{equation}
    Hence, if $W(\tu)\ge\max\{\gamma_\delta^2V(\tx), \frac{16M^2}{b^2}|\varepsilon_0\Pi(\theta)|^2\}$, we have that \eqref{ex1_usysbd} holds. Since $\gamma\gamma_\delta\in (0,1)$, {it also follows that $\gamma^2\gamma^2_\delta\in(0,1)$. Finally, we can define the functions $\hat{\gamma}_1(s):=\gamma^2 s$ and $\hat{\gamma}_2(s):=\gamma_\delta^2 s$.}
    It is straightforward to verify that $\hat{\gamma}_1\circ\hat{\gamma}_2(s)<s$ for all $s>0$. The conditions of Theorem \ref{thm_sgt} are satisfied, which implies that the interconnected system \eqref{ex1_transformed} is FxT-ISS with respect to the ``input" $\varepsilon_0 \Pi(\theta)$. This concludes the proof.
\end{proof}
\vspace{0.1cm}

The inequality \eqref{ex1sgt} can be interpreted as a measure on the ``relevance" of the input $u$ relative to the state $x$ in the $u$-subsystem \eqref{ex_fbk_u}. The Lipschitz constant $K$, defined in Assumption \ref{assump_fbk_wlip}, upper-bounds the influence that $x$ has on the ``quasi-gradient'' $\cg_\theta(x,u)$, while $\mu$ quantifies the level of strong convexity of the actual gradient $\cg_\theta(h(u),u)$. Since $\gamma$ is often increasing in the Lipschitz constant $\ell$, the parameter $\ell+\gamma$ can be viewed as a term that measures the trade-off between the required input-to-state gains in each subsystem. For instance, if $\ell+\gamma$ is large, meaning that the $x$-subsystem has a large $u$-to-$x$ gain, then, for the small-gain condition to hold, the parameter $\mu$ needs to be large relative to $K$, meaning that the state $u$ needs to be weighted more heavily than $x$ in the $u$-subsystem. Meanwhile, if $\ell+\gamma$ is small, then the $x$- subsystem more easily satisfies the FxT-ISS condition, which allows for more flexibility on $\mu$ relative to $K$ in the $u$-subsystem. This behavior is often expected when studying systems via small-gain theory.
\begin{remark}
     It can be observed that, as long as condition \eqref{ex1sgt} is satisfied, the proposed feedback optimization scheme \eqref{ex_fbk_sys} mitigates the timescale separation requirement imposed in the results of \cite{11239427}. An important consequence of this is that, contrary to the results in \cite{11239427}, the variations of $\theta$ in our setting no longer need to be measured relative to a highly conservative timescale separation parameter. Another important relaxation made in Theorem \ref{thm_fbk} relative to \cite{11239427} is that Theorem \ref{thm_fbk} does not make any assumptions on the degree of homogeneity in either subsystem. Indeed, the fixed-time feedback optimization scheme presented in \cite{11239427} requires that the homogeneity degree of the plant in the 0-limit be smaller than that of the controller, while the homogeneity degree of the plant in the  $\infty$-limit be larger than the corresponding homogeneity degree of the controller (see \cite{doi:10.1137/060675861}). Meanwhile, Theorem \ref{thm_fbk} circumvents this requirement by instead only requiring that \eqref{ex1sgt} holds.
     \QEDB
\end{remark}
\vspace{0.1cm}

\subsection{Numerical Example}
To illustrate Theorem \ref{thm_fbk} numerically, we simulate system \eqref{ex_fbk_sys} with plant dynamics \eqref{accplant} with $x,u\in\re^2$, $h(u)=1.1u$, $\tilde{p}=\frac12, \tilde{q}=-\frac12$, and $A_1=[1.5,0.3;0.3,1.8]$, $A_2=[1.4,0.25;0.25,1.6]$.
% \begin{equation*}
%     A_1=\begin{bmatrix}
%         1.5 & 0.3\\ 0.3 & 1.8
%     \end{bmatrix},\quad A_2=\begin{bmatrix}
%         1.4 & 0.25 \\ 0.25&1.6
%     \end{bmatrix}.
% \end{equation*}
It can be verified that this plant satisfies Assumption \ref{assump_ex} with $\gamma=\frac75$. For this example, the cost function takes the form
\begin{equation*}
    \phi_\theta(x,u)=\frac12 x^\top Q_\theta x+c^\top x+\frac12 u^\top P_\theta u+d^\top u,
\end{equation*}
where $Q_\theta, P_\theta\in\re^{2\times 2}$, $c,d\in\re^2$ are given by $Q_{\theta}=[1+d_1(t),0.5;0.5,1+d_2(t)]$, $P_{\theta}=[3.5,0.2;0.2,3.2+d_3(t)]$
%
% \begin{align*}
%     Q_\theta&=\begin{bmatrix}
%         1+d_1(t) & 0.5 \\ 0.5 & 1+d_2(t)
%     \end{bmatrix},\quad P_\theta=\begin{bmatrix}
%         3.5 & 0.2 \\ 0.2 & 3.2+d_3(t)
%     \end{bmatrix},
% \end{align*}
and $c=[2;1]$, $d=[1.5;0.5]$.
% \begin{equation*}
%     c=\begin{bmatrix}
%         2 \\ 1
%     \end{bmatrix},\quad d=\begin{bmatrix}
%         1.5 \\ 0.5
%     \end{bmatrix}.
% \end{equation*}
The disturbances are given by $d_1(t)=\frac12\sin(2\varepsilon_0 t), d_2(t)=\frac13\sin(1.5\varepsilon_0 t)$, and $d_3(t)=\frac17\sin(4\varepsilon_0 t)$. These signals can be generated by a system of the form \eqref{dtheta} by setting $\theta(t)\in\mathbb{R}^6$, $d_i(t)=\theta_{2i}(t)$, and $\Pi(\theta)=\mathcal{R}\theta$, where $\mathcal{R}\in\mathbb{R}^{6\times 6}$ is a block diagonal matrix with rotation matrices on the diagonal.
We also have
\begin{equation*}
    \cg_\theta(x,u)=1.1Q_\theta x+P_\theta u+1.1c+d.
\end{equation*}
For this choice of parameters, we obtain $\ell=1.1, K=1.339$, and $\mu=3.37$. Hence, condition \eqref{ex1sgt} is satisfied, and for any choice of $\xi_1\in (0,1)$ and $\xi_2<0$ we conclude that the closed-loop system is FxT ISS. The trajectories of the system, with $\xi_1=\frac13$ and $\xi_2=-\frac15$ and varying initial conditions, are shown in Figure \ref{ex0plot}.
\section{FIXED-TIME GRADIENT PLAY WITH DYNAMIC PLANTS IN NONCOOPERATIVE GAMES}
\label{sec_fixedtime}
In this section, we consider a more challenging problem compared to the standard feedback optimization setting studied in the previous section: fixed-time learning of Nash equilibrium in N-player noncooperative games with dynamic plants in the loop. {Such settings arise whenever self-interested agents compete over a shared physical system rather than a static payoff, as in power markets on an electrical grid, transmit-power control in wireless networks, and routing in transportation systems, where each agent's action affects a central plant \cite{6060862}.} In the asymptotic regime, and when the plant of each player $i$ is a static map $J_i(u)$ with suitable monotonicity properties, this problem can be solved using a standard pseudo-gradient flow of the form $\dot{u}_i=-\nabla_i J_i(u)$ \cite{rosen1965existence}. On the other hand, to achieve fixed-time convergence, \cite[Sec. IV-A]{9683248} introduced the \emph{decoupled fixed-time pseudogradient flow} $\dot{u}_i=-\fx(\nabla_iJ_i(u))$, where the function $\fx$ is given by \eqref{scale_function}, and which only requires for each player to have access to their own partial derivative $\nabla_iJ_i(u)$. In contrast to these settings, here we ask whether a similar FxT result can be obtained now for games with plants characterized by \emph{dynamical systems} \cite{romano2025game,belgioioso2024online}. We will tackle this problem using the FxT Small Gain Theorem studied in Section \ref{sec_main}.
\subsection{Model and Algorithm}
We consider an $N$-player noncooperative game, where each player $i\in [N]$ controls their action $u_i\in \re$, which maps to their payoff, $J_{i}(x,u)\in\re$, via a dynamic plant of the form:
\begin{align}\label{ex_nesplant}
    \dot{x}&=f(x, u),
\end{align}
where $x\in\re^n$ is the state, $u=[u_1,...,u_N]^\top$ is the vector of players' actions, and $f:\re^n\times\re^N\to\re^n$ is a continuous function. We make the following assumption on the plant \eqref{ex_nesplant}, which is similar to Assumption \ref{assump_ex} from Section \ref{sec_fbkopt}:

\vspace{0.1cm}
\begin{assumption}\label{assump_nes_h}
    For system \eqref{ex_nesplant}, there exists an $\ell$-Lipschitz, $\cc^1$ mapping $h:\re^N\to\re^n$ such that $f(h(u),u)=0$ for all $u\in\re^N$. Moreover, there exists $\gamma>0$ such that for each $\hat{u}\in\re^N$, there exists $a>0, p\in (0,1)$, and $q>1$ such that
    \begin{align}\label{ex_nes_fxtiss}
        &|x-h(\hat{u})|\ge\gamma |u-\hat{u}|\notag\\&\Rightarrow 2(x-h(\hat{u}))^\top f(x,u)\le -a|x-h(\hat{u})|^{2p}-a|x-h(\hat{u})|^{2q},
    \end{align}
    for all $x\in\re^n$ and $u\in\re^N$.\QEDB
\end{assumption}
\vspace{0.1cm}

The goal of player $i\in [N]$ is to design an update law on $u_i$ in a way that optimizes their payoff $J_{i}(x,u)$. Since the convergence of \eqref{ex_nesplant} requires $x=h(u)$, the players essentially want to seek the \emph{Nash equilibrium} of the game parameterized by the cost functions $\{\cj_{i}(u)\}_{i=1}^N$, where $\cj_{i}(u)=J_{i}(h(u),u)$. We will refer to this game as the \emph{quasi-steady state game}. The Nash equilibrium of this game is defined as the following:

\vspace{0.1cm}
\begin{definition}
    Consider an $N$-player game, where player $i\in[N]$ has action $u_i\in\re$ and cost function $\cj_{i}(u)$, where $u=[u_1,...,u_N]^\top$. An action profile $u^*\in\re^N$ is said to be a \emph{Nash Equilibrium} (NE) if, for each $i\in[N]$, the following holds
    \begin{equation*}
        \cj_i(u^*_1,...,u^*_{i-1},u_i, u^*_{i+1},...,u^*_N)\geq \cj_i(u^*),
    \end{equation*}
    for all $u_i\in\re$.
    \QEDB
\end{definition}
\vspace{0.1cm}

To seek the Nash equilibrium of the quasi-steady state, we introduce a decoupled fixed-time pseudogradient flow on the players' actions:
\begin{equation*}
    \dot{u}_i=-\fx(\nabla_{i} \cj_{i}(u)),
\end{equation*}
for some given $\xi_1\in (0,1)$ and $\xi_2<0$, where $\nabla_i$ denotes the partial derivative with respect to the $i$-th argument and $\fx$ is given in \eqref{scale_function}. It can be verified, via chain rule, that
\begin{align*}
    \nabla_{i} \cj_{i}(u)&=[H(u)^\top \nabla J_{i}(h(u),u)]_i\\
    &=\begin{bmatrix}
        \frac{\partial h(u)}{\partial u_i}^\top & e_i^\top
    \end{bmatrix}\nabla J_{i}(h(u),u),
\end{align*}
where $H(u)$ is given in \eqref{ex0ghu}, and $e_i$ is the $i$-th standard basis vector of $\re^N$. By replacing the steady-state approximation $h(u)$ with the measured value $x$, we obtain a real-time feedback control law on $u_i$:
\begin{equation}\label{pgf_decent}
    \dot{u}_i=-\fx(\cg_{i}(x,u)),
\end{equation}
where $\cg_{i}$ is given by
\begin{equation*}
    \cg_{i}(x,u)=\begin{bmatrix}
        \frac{\partial h(u)}{\partial u_i}^\top & e_i^\top
    \end{bmatrix}\nabla J_{i}(x,u).
\end{equation*}
We then obtain the following interconnected system:
\begin{subequations}\label{ex_nes_int}
    \begin{align}
        \dot{x}&=f(x,u)\\
        \dot{u}_i&=-\fx(\cg_{i}(x,u)),\quad i=1,...,N
    \end{align}
\end{subequations}
We now introduce the following operator
\begin{equation*}
    \cg(x,u)=[\cg_{1}(x,u),...,\cg_{N}(x,u)]^\top.
\end{equation*}
\begin{figure}[t!]
  \centering \includegraphics[width=0.45\textwidth]{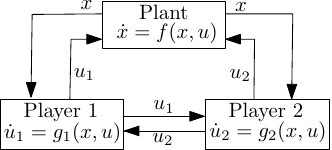}
    \caption{A diagram depicting the interconnection \eqref{ex_nes_int} for a two player game with a plant in the loop. } \label{nes_block}
    %\vspace{-0.4cm}
\end{figure}
The function $\cg(x,u)$ can be seen as a feedback-based version of the \emph{pseudogradient} operator that is found in traditional Nash equilibrium seeking algorithms \cite{yi2019operator,11048686}. In fact, it follows directly by definition that $\cg(h(u),u)$ is the pseudogradient (in the traditional sense) of the quasi-steady state game. 

In this paper, we focus on a class of games termed \emph{potential games}, which are characterized by the following assumption:

\vspace{0.1cm}
\begin{assumption}\label{assump_p}
There exists a $\cc^1$ function $P:\re^N\to\re$ such that $\nabla P(u)=\cg(h(u),u)$ for all $u\in\re^N$.
Moreover, the following holds:
    \begin{enumerate}
        \item The Nash equilibrium of the quasi-steady state game is unique and satisfies $u^*=\argmin_u P(u)$.
        \item The function $P$ is $\mu$-strongly convex, i.e., there exists $\mu>0$ such that the following holds
            \begin{align}\label{assump_p_convex}
                P(\hat{u})\ge P(u)+\nabla P(u)^\top (\hat{u}-u)+\frac{\mu}{2}|\hat{u}-u|^2,
            \end{align}
            for all $u, \hat{u}\in\re^N$.
        \item There exists $K>0$ such that
        \begin{equation*}
            |\cg(\hat{x},u)-\cg(x,u)|\le K|\hat{x}-x|,
        \end{equation*}
        for all $x,\hat{x}\in\re^n$ and $u\in\re^N$. \QEDB
    \end{enumerate}
\end{assumption}

\vspace{0.1cm}
\begin{remark}
    The function $P$ is referred to as the \emph{potential function}. Games that admit potential functions are referred to as \emph{potential games}, which are ubiquitous in the game theory literature \cite{marden2017game, 6060862}. Assumption \ref{assump_p} can be viewed as analogous to Assumptions \ref{assump_phi} and \ref{assump_fbk_wlip} in Section \ref{sec_fbkopt} for fixed $\theta$, except that we now do not require $P$ to have a Lipschitz gradient. It is important to note that under Assumption \ref{assump_p}, the interconnected dynamics \eqref{ex_nes_int} are not equivalent to the dynamics \eqref{ex_fbk_sys} studied in Section \ref{sec_fbkopt}. In particular, the dynamics \eqref{ex_nes_int} of each player $i$ implement  normalizing terms that only depend on the $i^{th}$ player's partial derivative $\nabla_i J_i$, which is essential for a decentralized implementation.
    \QEDB
\end{remark}
\vspace{0.1cm}

We are now ready to state the main result of this section, which establishes FxTS of the Nash equilibrium, $[h(u^*), u^*]^\top$, for the interconnected system \eqref{ex_nes_int}.

\vspace{0.1cm}
\begin{thm}\label{thm_nes}
    Suppose that Assumptions \ref{assump_nes_h}-\ref{assump_p} hold, as well as the following condition:
    \begin{equation}\label{nes_sgt}
        \mu>\tcb{\sqrt{N}}K(\ell+\gamma).
    \end{equation}
    Then, there exists $\beta\in\klfx$ such that, for each $x(0)\in\re^n$ and $u(0)\in\re^N$, each solution to \eqref{ex_nes_int} satisfies
    \begin{align*}
        &\left\lvert\begin{bmatrix}
            x(t)-h(u^*)\\ u(t)-u^*
        \end{bmatrix}\right\rvert\le \beta\left(\left\lvert\begin{bmatrix}
            x(0)-h(u^*)\\ u(0)-u^*
        \end{bmatrix}\right\rvert,t\right)
    \end{align*}
    for all $t\ge 0$.\QEDB
\end{thm}
\begin{proof}
    With the transformation $\tx=x-h(u^*), \tu=u-u^*$, we arrive at the following system
    \begin{subequations}\label{ex_nes_transf}
    \begin{align}
        \dot{\tx}&=f(\tx+h(u^*),\tu+u^*)\\
        \dot{\tu}_i&=-\fx(\cg_{i}(\tx+h(u^*),\tu+u^*)),\quad i=1,...,N.
    \end{align}
\end{subequations}
First, note that \eqref{assump_p_convex} implies $P(\tu+u^*)-P(u^*)\ge \frac{\mu}{2}|\tu|^2$. For the $\tx$- and $\tu$-subsystems, consider the FxT-ISS Lyapunov functions $V(\tx)=|\tx|^2$ and $W(\tu)=P(\tu+u^*)-P(u^*)$, respectively. Let $\tilde{\gamma}^2=\frac{2}{\mu}\gamma^2$.
From Assumption \ref{assump_nes_h}, we obtain the following implication
\begin{align*}
    &V(\tx)\ge \tilde{\gamma}^2W(\tu)\Rightarrow |\tx|^2\ge \gamma^2 |\tu|^2\\
    &\Rightarrow \dot{V}(\tilde{x})=2\tx^\top f(\tx+h(u^*),\tu+u^*)\le -aV^p(\tx)-aV^q(\tx).
\end{align*}
Hence, the $\tx$-subsystem satisfies the Lyapunov-based FxT-ISS implication \eqref{fxtiss_imp_assump} with $\gamma_1(s)=\tilde{\gamma}^2 s$. Since $P$ is strongly convex, it also satisfies the following bounds:
\begin{subequations}
    \begin{align}
        |\nabla P(\tu+u^*)|&\ge \mu |\tu|\\
        P(\tu+u^*)-P(u^*)&\le \frac{1}{2\mu}|\nabla P(\tu+u^*)|^2\label{plineq},
    \end{align}
\end{subequations}
for all $\tu\in\re^N$. For the $\tu$-subsystem, let us denote
\begin{align*}
    G_i(\tu)&:=\cg_i(h(\tu+u^*), \tu+u^*)\\
    \ov{G}_i(\tx,\tu)&:=\cg_i(\tx+h(u^*),\tu+u^*)-\cg_i(h(\tu+u^*),\tu+u^*),
\end{align*}
with $G(\tu)=[G_1(\tu),...,G_N(\tu)]^\top$ and $\ov{G}(\tx,\tu)=[\ov{G}_1(\tx,\tu),...,\ov{G}_N(\tx,\tu)]^\top$. Note that $G(\tu)=\nabla P(\tu+u^*)$. By condition \eqref{nes_sgt}, we can pick $\varepsilon>0$ such that
\begin{equation*}
    \mu>\tcb{\sqrt{N}}K(\ell+\varepsilon+\gamma).
\end{equation*} Suppose that $W(\tu)\ge \delta^2|\tx|^2$, where $\delta^2$ is given by
\begin{equation*}
    \delta^2={\frac{\mu}{2}}\frac{1}{(\gamma+\varepsilon)^2}.
\end{equation*}
We then have
\begin{align}
    |\ov{G}(\tx,\tu)|&\le K|\tx-(h(\tu+u^*)-h(u^*))|\notag\\
    &\le K(|\tx|+\ell |\tu|)\notag \\&\le 
    K\left(\frac{1}{\delta\sqrt{2\mu}}+\frac{\ell}{\mu}\right)|\nabla P(\tu+u^*)|\notag\\
    &=\frac{K(\ell+\varepsilon+\gamma)}{\mu}|G(\tu)|.
    \label{ex_nes_gbd}
\end{align}
Note that $\frac{K(\ell+\varepsilon+\gamma)}{\mu}<\tcb{\frac{1}{\sqrt{N}}}$, hence by Lemma \ref{lem_decentralized} in the Appendix there exists $c>0$ such that:
\begin{align*}
    \dot{W}(\tu)&=-\sum_{i=1}^N G_i(\tu)\fx(G_i(\tu)+\ov{G}_i(\tx,\tu))\\
    &\le -c|G(\tu)|^{2-\xi_1}-c|G(\tu)|^{2-\xi_2}\\
    &\le -c(2\mu)^{1-\frac12\xi_1}W^{1-\frac12\xi_1}(\tu)-c(2\mu)^{1-\frac12\xi_2}W^{1-\frac12\xi_2}(\tu),
\end{align*}
where the bottom inequality uses the fact that $|\nabla P(\tu+u^*)|\ge \sqrt{2\mu}W^\frac12 (\tu),$
which follows directly from \eqref{plineq}.
Since 
\begin{equation*}
    \delta^2\tilde{\gamma}^2=\left(\frac{\gamma}{\gamma+\varepsilon}\right)^2<1,
\end{equation*}
{we can choose $\hat{\gamma}_1(s)=\tilde{\gamma}^2s$ and $\hat{\gamma}_2(s)=\delta^2s$ and observe that the small gain condition of Theorem \ref{thm_sgt} is satisfied}. This establishes the result.
\end{proof}
\vspace{0.1cm}
\begin{remark}
    The condition \eqref{nes_sgt} is highly reminiscent of condition \eqref{ex1sgt}, but we now require the lower bound on $\mu$ to be proportional to the square root of the number of players in the game. Hence, as the player count increases, $u$ needs to be weighted much more heavily than $x$ in the pseudogradient $\cg(x,u)$. This is unsurprising, since we are primarily concerned with players implementing the decoupled pseudogradient flow \eqref{pgf_decent}, where it is typically the case that the dimension should be accounted for \cite{9760031}. Meanwhile, if the players choose to implement a \emph{centralized} fixed-time pseudogradient flow, then the problem can be studied by emulating the results from Section \ref{sec_fbkopt}.
    \QEDB
\end{remark}
\vspace{0.1cm}

\begin{figure}[t!]
  \centering \includegraphics[width=0.45\textwidth]{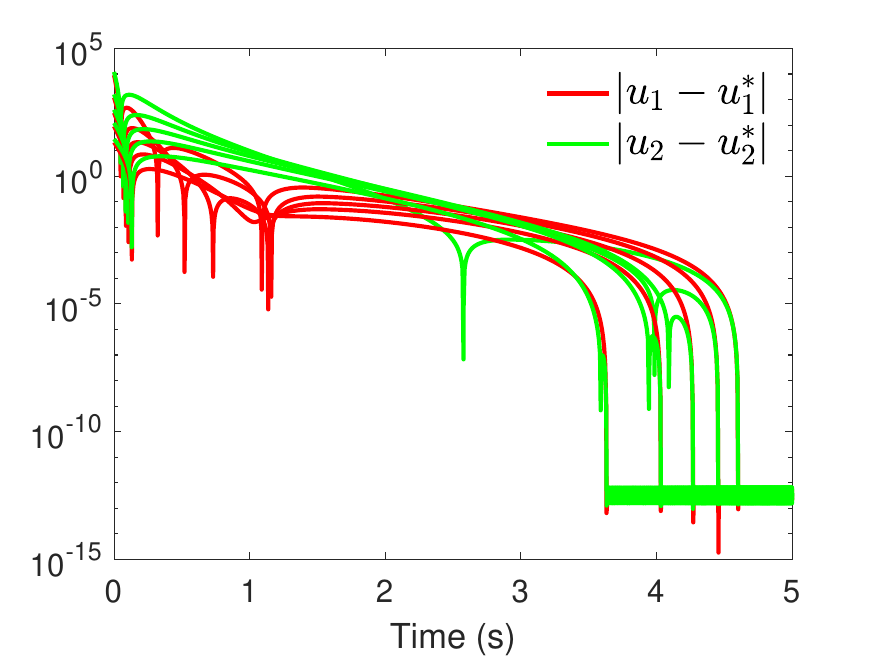}
    \caption{Trajectories showcasing the convergence of the players' actions in \eqref{ex_nes_int} for a two player game.} \label{nes_plot}
    %\vspace{-0.4cm}
\end{figure}
\subsection{Numerical Examples}
%
%We illustrate Theorem \ref{thm_nes} by presenting a numerical example. 
Consider a 2-player game, where the cost function of the $i^{th}$ player is given by
\begin{equation*}
    J_i(x,u)=\frac12 x^\top Q_i x+c_i^\top x+\frac12 u^\top P_i u+d_i^\top u,
\end{equation*}
where $Q_1=[1,0.5;0.5,0.75]$, $P_1=[4,1;1,2]$, $c_1=[0.6;1.3]$, $d_1=[0.7;0.5]$, $Q_2=[2,1.7;1.7,1.5]$, $P_2=[2,0.7;0.7,3]$, $c_2=[1.5;1]$, $d_2=[0.3;0.8]$, where we used so-called MATLAB notation.
%
% \begin{align*}
%     Q_1&=\begin{bmatrix}
%         1 & 0.5\\ .5 &0.75
%     \end{bmatrix},\ P_1=\begin{bmatrix}
%         4&1\\ 1&2
%     \end{bmatrix},\ c_1=\begin{bmatrix}
%         0.6\\1.3
%     \end{bmatrix},\ d_1=\begin{bmatrix}
%         0.7\\0.5
%     \end{bmatrix}\\
%     Q_2&=\begin{bmatrix}
%         2&1.7\\ 1.7&1.5
%     \end{bmatrix},\ P_2=\begin{bmatrix}
%         2&.7\\.7 &3
%     \end{bmatrix},\ c_2=\begin{bmatrix}
%         1.5\\1
%     \end{bmatrix},\ d_2=\begin{bmatrix}
%         0.3\\0.8
%     \end{bmatrix}.
% \end{align*}
Moreover, the plant is given by \eqref{accplant} with $A_1=A_2=\mathbb{I}_2$ and $h(u)=\frac12 u$. For this choice of parameters, we have $\ell=\frac12, \gamma=0.51$, and $\cg(x,u)=\frac12Qx+Pu+\frac12c+d,$
where $Q=[1,0.5;1.7,1.5]$, $P=[4,1;0.7,3]$, $c=[0.6;1]$, $d=[0.7;0.8]$.
%
% \begin{equation*}
%     Q=\begin{bmatrix}
%         1 &.5\\ 1.7 &1.5
%     \end{bmatrix},\ P=\begin{bmatrix}
%         4&1\\.7&3
%     \end{bmatrix},\ c=\begin{bmatrix}
%         .6\\1
%     \end{bmatrix},\ d=\begin{bmatrix}
%         .7\\.8
%     \end{bmatrix}.
% \end{equation*}
%
It can be verified that the constrained game is a potential game, with potential function
\begin{equation*}
    P(u)=\frac12 u^\top \mathcal{Q}u+\left(\frac12c+d\right)^\top u,\quad \mathcal{Q}=\begin{bmatrix}
        4.25 & 1.125 \\ 1.125 & 3.375
    \end{bmatrix}.
\end{equation*}
We have $K=1.257$ and $\mu=2.605$, so condition \eqref{nes_sgt} is satisfied. Hence, by Theorem \ref{thm_nes}, we can conclude that the point $[h(u^*), u^*]$ for the interconnected NES dynamics is FxTS. The trajectories of $|u_i-u_i^*|$ with $\tilde{p}=\frac25$, $\tilde{q}=-\frac27, \xi_1=\frac13, \xi_2=-\frac15$ are shown in Figure \ref{nes_plot}.

\section{CONCLUSION}
\label{sec_conclusions}
We introduced Lyapunov-based small-gain conditions for establishing fixed-time input-to-state stability (FxT ISS) in interconnected dynamical systems. To the authors' knowledge, this is the first Lyapunov-based result of its kind in the literature on fixed-time stability. The theoretical developments are demonstrated through several illustrative examples. The first example, a second-order interconnection, provides a clear paradigm for applying the proposed tools to study interconnections of FxTS systems with homogeneous disturbances. In the second example, we apply our results to analyze fixed-time feedback optimization in dynamical systems without requiring time-scale separation between the plant and the controller. Finally, we also show that the proposed results can be used to study equilibrium-seeking problems in potential noncooperative games with dynamic plants in the loop and decoupled FxT gradient-based dynamics. Future work includes extending the results to large-scale networked systems, relaxing some of the structural assumptions on the gain functions, and {exploring further applications in engineering and sciences.}
%

%\vspace{-0.5cm}
\section{Appendix}
\begin{lemma}\label{lemdom}
    Let $\sigma(s)=\sum_{i=1}^n c_i s^{p_i}$ satisfy $\sigma(s)>0$ for $s>0$, where $c_i\neq 0$ for each $i\in[n]$ and $p_1<p_2<...<p_n$. Then, there exists $\ov{c}, \un{c}>0$ such that
    \begin{equation*}
        \un{c}(s^{p_1}+s^{p_n})\le\sigma(s)\le \ov{c}(s^{p_1}+s^{p_n}),
    \end{equation*}
    for all $s> 0$.\QEDB
\end{lemma}
\begin{proof}
Since $\sigma(s)>0$ for all $s>0$, we must have that $c_1>0$ and $c_n>0$. Consider the function $\tilde{\sigma}:(0, \infty)\to (0, \infty)$ defined by
    $\tilde{\sigma}(s)=\frac{\sigma(s)}{s^{p_1}+s^{p_n}}$. It is straightforward to verify that $\lim_{s\to 0^+}\tilde{\sigma}(s)=c_1$, so we can pick $\varepsilon_1>0$ such that $\tilde{\sigma}(s)\in(\frac{c_1}{2}, \frac{3c_1}{2})$ for all $s\in (0, \varepsilon_1)$. Moreover, since $\lim_{s\to \infty}\tilde{\sigma}(s)=c_n$, we can pick $\varepsilon_2>\varepsilon_1$ such that $\tilde{\sigma}(s)\in (\frac{c_n}{2}, \frac{3c_n}{2})$ for all $s>\varepsilon_2$. Since $\tilde{\sigma}$ is positive and continuous for $s>0$, there exists $\ov{\sigma}, \un{\sigma}>0$ such that $\un{\sigma}\le\tilde{\sigma}(s)\le \ov{\sigma}$ for all $s\in [\varepsilon_1, \varepsilon_2]$. We can then pick $\un{c}=\min\{\un{\sigma}, \frac{c_1}{2}, \frac{c_n}{2}\}$ and $\ov{c}=\max\{\ov{\sigma}, \frac{3c_1}{2}, \frac{3c_n}{2}\}$ to complete the proof.
\end{proof}

\vspace{0.1cm}

\begin{lemma}\label{lem_k1scale}
    Let $\Psi\in\ckf$ and $\sigma(s)=\sum_{i=1}^n c_i s^{r_i}$, where $1\le r_1< r_2< ..< r_n$ and $c_i\neq 0$ for $i\in [n]$. Moreover, suppose $\sigma'(s)>0$ for all $s>0$. Then, there exists $\tilde{\Psi}\in\ckf$ such that $\tilde{\Psi}(\sigma(s))\le\Psi(s)\sigma'(s)$ for all $s>0$.\QEDB
\end{lemma}

\vspace{0.1cm}
\begin{proof}
    Let $\Psi(s)=a s^{p}+bs^q$, where $p\in (0,1)$, $q>1$, and $a,b>0$. It suffices to consider $n\ge 3$. We have $\sigma'(s)=\sum_{i=1}^n r_i c_i s^{r_i-1}$, which yields
    \begin{align*}
    \Psi(s)\sigma'(s)&=\sum_{i=1}^n ar_i c_i s^{p+r_i-1}+\sum_{j=1}^n br_j c_j s^{q+r_j-1}\notag\\&\ge {c}s^{p+r_1-1}+{c}s^{q+r_n-1},
    % +\left(\frac{m}{2}s^{p+r_1-1}+\frac{m}{2}s^{q+r_n-1}+\sum_{i=2}^{n-1} ar_i c_i s^{p+r_i-1}+\sum_{j=2}^{n-1} br_j c_j s^{q+r_j-1}\right),
    \end{align*}
    where ${c}$ is obtained from Lemma \ref{lemdom}. We then postulate $\tilde{\Psi}\in\ckf$ of the form $\tilde{\Psi}(s)=\varepsilon s^{\tilde{p}}+\varepsilon s^{\tilde{q}}$, where $\tilde{p}=\frac{p+r_1-1}{r_1}\in (0,1)$ and $\tilde{q}=\frac{q+r_n-1}{r_n}>1$. Indeed, computations yield:
    \begin{align*}
        &\tilde{\Psi}(\sigma(s))=\varepsilon\left(\sum_{i=1}^n c_i s^{r_i}\right)^{\tilde{p}}+\varepsilon\left(\sum_{j=1}^n c_j s^{r_j}\right)^{\tilde{q}}\\
        &~~~~\le \varepsilon\left(\sum_{i=1}^n |c_i|^{\tilde{p}}s^{r_i \tilde{p}}+n^{\tilde{q}-1}\sum_{j=1}^n |c_j|^{\tilde{q}}s^{r_j\tilde{q}}\right)\\
    &~~~~\le \varepsilon\left(\sum_{i=1}^n |c_i|^{\tilde{p}}+n^{\tilde{q}-1}\sum_{j=1}^n |c_j|^{\tilde{q}}\right)(s^{p+r_1-1}+s^{q+r_n-1}),
    \end{align*}
    where the middle inequality follows from Lemma \ref{jensenlemma} and the bottom inequality follows from Lemma \ref{lem_sandw} with the fact that $r_i\tilde{p}, r_i \tilde{q}\in [p+r_1-1, q+r_n-1]$ for all $i\in [n]$. We can then pick
    \begin{equation*}
        \varepsilon\le \frac{{c}}{\sum_{i=1}^n |c_i|^{\tilde{p}}+n^{\tilde{q}-1}\sum_{j=1}^n |c_j|^{\tilde{q}}},
    \end{equation*}
    to obtain
    \begin{equation*}
        \tilde{\Psi}(\sigma(s))\le {c}s^{p+r_1-1}+{c}s^{q+r_n-1}\le \Psi(s)\sigma'(s),
    \end{equation*}
    which yields the result.
\end{proof}
\vspace{0.1cm}
\begin{lemma}\label{lem_invscale}
    Suppose $\gamma\in\cki$, and let $\Psi\in\ckf$. Then there exists $\lambda\ge 1$ and $\tilde{\Psi}\in\ckf$ such that $\lambda\gamma^{\lambda-1}(s)\gamma'(s)\Psi(s)\ge \tilde{\Psi}(\gamma^\lambda(s))$ for all $s> 0$.\QEDB
\end{lemma}
\begin{proof}
    Denote $\Psi(s)=a s^{p}+b s^{q}$, where $p\in (0,1)$ and $q>1$ and let $\gamma^{-1}(s)=\sum_{i=1}^n c_i s^{r_i}$, where $0<r_1<r_2<...<r_n$ and $c_i\neq 0$ for all $i\in [n]$. 
    % Moreover, let $C=\sum_{i=1}^n c_i$. First, we consider the case where $s\in (0,1)$, which implies $\gamma^{-1}(s)\le Cs^{r_1}$, and thus $\gamma(s)\ge C^{-\frac{1}{r_1}}s^\frac{1}{r_1}$. 
    From Lemma \ref{lemdom} we can pick $c>0$ such that $\gamma^{-1}(s)\ge c(s^{r_1}+s^{r_n})$. It follows that $\gamma^{-1}(s)\ge cs^{r_1}$ and  $\gamma^{-1}(s)\ge cs^{r_n}$. Since these inequalities involve class-$\mathcal{K}_{\infty}$ functions, it also follows that $\gamma(s)\leq \tilde{c}_k s^\frac{1}{r_k}$ for all $k\in\{1,n\}$ and $s\ge 0$, where $\tilde{c}_k=c^{-\frac{1}{r_k}}$ and where the upper bound is the inverse function of $cs^{r_k}$. It follows that $\Psi(s)\ge a c^{p}\gamma^{p r_k}(s)+b c^{q}\gamma^{q r_k}(s)$ for $k\in\{1,n\}$. Let $R:=\sum_{i=1}^n r_i |c_i|$, and consider the following two cases:

    \vspace{0.1cm}\noindent 
    \textsl{Case 1:} $\gamma(s)\in (0,1)$, which implies
    \begin{equation*}
        \sum_{i=1}^n r_i c_i \gamma^{r_i-1}(s)\le \sum_{i=1}^n r_i |c_i|\gamma^{r_1-1}(s)= R\gamma^{r_1-1}(s).
    \end{equation*}
    Denote $\varrho_\lambda(s):=\lambda\gamma^{\lambda-1}(s)\gamma'(s)\Psi(s)$. By the inverse function theorem, we have 
    \begin{equation}\label{derivativegamma}
    \gamma'(s)=\frac{1}{(\gamma^{-1})'(\gamma(s))},
    \end{equation}
    for $s>0$,
    which yields
    \begin{align*}
        \varrho_\lambda(s)&=\frac{\lambda\gamma^{\lambda-1}(s)\Psi(s)}{\sum_{i=1}^n r_i c_i \gamma^{r_i-1}(s)}\\&\ge \frac{\lambda a c^{p}\gamma^{\lambda-1+p r_k}(s)+\lambda b c^{q}\gamma^{\lambda-1+qr_k}(s)}{R\gamma^{r_1-1}(s)}\\
        &\ge \varepsilon\gamma^{\lambda+p r_k-r_1}(s)+\varepsilon\gamma^{\lambda+q r_k-r_1}(s),\quad k\in\{1,n\}
    \end{align*}
    where $\varepsilon=\frac{\lambda}{R}\min\left\{ac^{p}, b c^{q}\right\}$. Choosing $k=1$ yields $\varrho_\lambda(s)\ge \tilde{\Psi}_1(\gamma^\lambda(s))$,
where $\tilde{\Psi}_k$ is given by
\begin{equation*}
    \tilde{\Psi}_k(s):=\varepsilon s^{1-\frac{1}{\lambda}{(1-p)r_k}}+\varepsilon s^{1+\frac{1}{\lambda}{(q-1)r_k}}.
\end{equation*}
Picking $\lambda>(1-p)r_1$ guarantees $\tilde{\Psi}_1\in\ckf$. 

\vspace{0.1cm}\noindent 
    \textsl{Case 2:} $\gamma(s)\ge 1$, which implies $\sum_{i=1}^n r_i c_i \gamma^{r_i-1}(s)\le R\gamma^{r_n-1}(s)$. We can then perform similar calculations done above to obtain:
\begin{equation*}
    \varrho_\lambda(s)\ge \varepsilon\gamma^{\lambda+p r_k-r_n}(s)+\varepsilon\gamma^{\lambda+q r_k-r_n}(s),\quad k\in \{1,n\}.
\end{equation*}
If we choose $k=n$, we obtain $\varrho_\lambda(s)\ge \tilde{\Psi}_n(\gamma^\lambda(s))$,
where $\tilde{\Psi}_n\in\ckf$ if $\lambda>(1-p)r_n$. 

\vspace{0.1cm}\noindent 
Since we assume $r_1<r_n$, it suffices to have $\lambda>(1-p)r_n$. By Lemma \ref{lem_fxtbd}, we know that there exists $\tilde{\Psi}\in\ckf$ such that $\min(\tilde{\Psi}_1(s), \tilde{\Psi}_n(s))\ge \tilde{\Psi}(s)$ for all $s>0$. By combining both cases, we obtain
\begin{equation*}
    \varrho_\lambda(s)\ge \min\left\{\tilde{\Psi}_1(\gamma^\lambda(s)), \tilde{\Psi}_n(\gamma^\lambda(s))\right\}\ge \tilde{\Psi}(\gamma^\lambda(s)),
\end{equation*}
which establishes the result.
\end{proof}

\vspace{0.1cm}
\begin{lemma}\label{lem_k1fxt}
    Suppose $V:\re^n\to\re_{\ge 0}$ is a FxT-ISS Lyapunov function for system \eqref{sysfxts}. Let $\sigma\in\ck_\infty$ take the form $\sigma(s)=\sum_{i=1}^n c_i s^{r_i}$, where $1\le r_1< r_2< ..< r_n$ and $c_i\neq 0$ for $i\in [n]$. Moreover, suppose $\sigma'(s)>0$ for all $s>0$. It follows that $\tilde{V}(x)=\sigma(V(x))$ is a FxT-ISS Lyapunov function for system \eqref{sysfxts}.\QEDB
\end{lemma}
\begin{proof}
    First, it is easy to verify that $\sigma(\underline{\alpha}(|x|))\le \tilde{V}(x)\le \sigma(\overline{\alpha}(|x|)).$ Let $\tilde{\chi}=\sigma\circ \chi$, we then have
    \begin{equation*}
        \tilde{V}(x(t))\ge \tilde{\chi}(|u(t)|)\Rightarrow D^+ V(x(t))\le -\Psi(V(x(t))),
    \end{equation*}
    along the trajectories of \eqref{sysfxts} for all $t\ge 0$. Let $t\ge 0$ such that $V(x(t))>0$ and suppose $\tilde{V}(x(t))\ge \tilde{\chi}(|u(t)|)$, which implies
    \begin{align*}
        D^+\tilde{V}(x(t))&\le \sigma'(V(x(t)))D^+V(x(t))\\
        &\le -\sigma'(V(x(t)))\Psi(V(x(t))).
    \end{align*}
    By Lemma \ref{lem_k1scale}, we can pick $\tilde{\Psi}\in\ckf$ such that $\tilde{\Psi}(\sigma(s))\le \sigma'(s)\Psi(s)$ for all $s>0$, which implies $D^+\tilde{V}(x(t))\le -\tilde{\Psi}(\tilde{V}(x(t))).$
    Next, let $t\ge 0$ such that $V(x(t))=0$, which implies $\tilde{V}(x(t))=0$. Then, if $\tilde{V}(x(t))\ge \tilde{\chi}(|u(t)|)$ holds, we have $|u(t)|=0$ and $D^+V(x(t))\le 0$, which yields
    \begin{equation*}
        D^+\tilde{V}(x(t))\le\sigma'(0) D^+ V(x(t))\le 0.
    \end{equation*}
    This establishes the result.
\end{proof}
\vspace{0.1cm}
\begin{lemma}\label{lem_invfxt}
     Suppose $V:\re^n\to\re_{\ge 0}$ is a FxT-ISS Lyapunov function for system \eqref{sysfxts}, and let $\gamma\in\cki$. Then, there exists $\lambda\ge 1$ such that
     $\tilde{V}(x)=\gamma^\lambda(V(x))$
     is a FxT-ISS Lyapunov function for system \eqref{sysfxts}.\QEDB
\end{lemma}
\begin{proof}
    Let $\sigma(s)=\gamma^\lambda(s)$ and $\gamma^{-1}(s)=\sum_{i=1}^n c_i s^{r_i}$, where $0<r_1<r_2<...<r_n$ and $c_i\neq 0$ for all $i\in [n]$. First, we show that $\sigma$ is $\mathcal{C}^1$ on $[0, \infty)$. Indeed, by Lemma \ref{lemdom}, there exists $c>0$ such that $\gamma^{-1}(s)\ge c(s^{r_1}+s^{r_n})$ for all $s\ge 0$. Since $\gamma^{-1}$ is $\cc^1$ with a nonzero derivative on $(0, \infty)$, it follows from the inverse function theorem that $\sigma$ is also $\cc^1$ on $(0, \infty)$, and, using \eqref{derivativegamma}, its derivative is given by 
    \begin{equation}\label{derivatiegammaproof}
        \sigma'(s)=\frac{\lambda\gamma^{\lambda-1}(s)}{(\gamma^{-1})'(\gamma(s))}=\frac{\lambda\gamma^{\lambda-1}(s)}{\sum_{i=1}^n r_i c_i \gamma^{r_i-1}(s)},
    \end{equation}
    for all $s>0$.
    Moreover, note that $\gamma(s)\le \tilde{c}_k s^\frac{1}{r_k}$ for each $k\in \{1,n\}$ and $s\ge 0$, where $\tilde{c}_k=c^{-\frac{1}{r_k}}$. For $\lambda>r_n>r_1$, we obtain
    \begin{align*}
        0\le \sigma^+(0):=\lim_{s\to 0^+}\frac{\sigma(s)-\sigma(0)}{s}\le
        \lim_{s\to 0^+}\frac{\tilde{c}^\lambda_1 s^\frac{\lambda}{r_1}}{s}=0,
    \end{align*}
    which results in $\sigma^+(0)=0$. Since the derivative of $\gamma^{-1}$ is a positive definite polynomial, we can use Lemma \ref{lemdom} to obtain $\tilde{c}>0$ such that $(\gamma^{-1})'(\gamma(s))\ge \tilde{c}\gamma^{r_1-1}(s)$ for all $s> 0$. For $\lambda>r_1$, we also have
    \begin{equation*}
        0\le\lim_{s\to 0^+}\frac{\lambda\gamma^{\lambda-1}(s)}{\sum_{i=1}^n r_i c_i \gamma^{r_i-1}(s)}\le\lim_{s\to 0^+}\frac{\lambda}{\tilde{c}}\gamma^{\lambda-r_1}(s)=0,
    \end{equation*}
    which implies that $\sigma$ is $\mathcal{C}^1$ on $[0, \infty)$. 

    \vspace{0.1cm}
    Next, suppose $\tilde{V}(x)\ge \tilde{\chi}(|u|)$, where $\tilde{\chi}=\sigma\circ\chi$. This implies $V(x)\ge \chi(|u|)$, and thus there exists $\Psi\in\ckf$ such that $D^+ V(x(t))\le -\Psi(V(x(t)))$ along the trajectories of \eqref{sysfxts} for all $t\ge 0$. Let $t>0$ be such that $V(x(t))>0$. Performing similar computations as in the proof of Lemma \ref{lem_k1fxt}, we obtain:
    \begin{align*}
        D^+\tilde{V}(x(t))\le \sigma'(V(x(t)))D^+ V(x(t))\le -\varrho_\lambda(V(x(t))),
    \end{align*}
    where
    \begin{equation*}
        \varrho_\lambda(s)=\frac{\lambda\gamma^{\lambda-1}(s)\Psi(s)}{\sum_{i=1}^n r_i c_i \gamma^{r_i-1}(s)}=\lambda\gamma^{\lambda-1}(s)\gamma'(s)\Psi(s),
    \end{equation*}
    where the last equality follows by \eqref{derivativegamma} and \eqref{derivatiegammaproof}. Thus, by Lemma \ref{lem_invscale}, there exists $\lambda\geq1$ and $\tilde{\Psi}\in\ckf$ such that 
    $\varrho_\lambda(s)\ge \tilde{\Psi}(\gamma^\lambda(s))$ for all $s>0$. Therefore, we obtain:
    \begin{equation*}
        D^+\tilde{V}(x(t))\le -\tilde{\Psi}(\gamma^\lambda(V(x(t))))=-\tilde{\Psi}(\tilde{V}(x(t))).
    \end{equation*}
    Now pick $t\ge 0$ such that $V(x(t))=0$, which implies $\tilde{V}(x(t))=0$. If $\tilde{V}(x(t))=0$, then $\tilde{V}(x(t))\ge \tilde{\chi}(|u(t)|)$ implies $u(t)=0$, and hence $D^+V(x(t))\le 0$. We then have
    \begin{align*}
        D^+\tilde{V}(x(t))&\le \sigma'(0)D^+ V(x(t))\le 0,
    \end{align*}
    which establishes the result.
    % Now we consider the case of $V(x(t))=0$, which implies $\tilde{V}(x(t))=0$, $u(t)=0$ and $D^+V(x(t))\le 0$. Then, using the fact that $\gamma(s)\le \tilde{c}_i s^\frac{1}{r_i}$ for all $i\in[n]$, where $\tilde{c}_i=c_i^{-\frac{1}{r_i}}$, we have:
    % \begin{align}
    %     D^+\gamma^\lambda(0)&=\limsup_{h\to 0^+}\frac{\gamma^\lambda(h)-\gamma^\lambda(0)}{h}\\&\le \limsup_{h\to 0^+}\frac{\tilde{c}^\lambda_i h^\frac{\lambda}{r_i}}{h}
    % \end{align}
    % Picking $i=1$ and $\lambda>r_1$ yields $D^+\gamma^\lambda(0)=0$. We can follow the same arguments from \cite{sgfinite} to obtain
    % \begin{equation}
    %     D^+\tilde{V}(x(t))\le D^+\gamma^\lambda(0) D^+V(x(t))= 0
    % \end{equation}
    % Thus, $\tilde{V}=\gamma^\lambda\circ V$ is also a FxT ISS Lyapunov function for \eqref{sysfxts}.
\end{proof}

\vspace{0.1cm}
\begin{lemma}\label{lem_maxinv}
    Let $\alpha_1,...,\alpha_N\in\ck_\infty$ and $\alpha(s):=\max_k \alpha_k(s)$. Then, $\alpha^{-1}(s)=\min_k \alpha_k^{-1}(s)$.\QEDB
\end{lemma}
\begin{proof}
    We already know that $\alpha\in\ck_\infty$, so its inverse, $\alpha^{-1}$, is well-defined and also of class $\ck_\infty$. Fix $s>0$, and pick $r>0$ such that $\alpha(r)=s$. We have that $\alpha(r)=\alpha_i(r)$ for some $i\in [N]$, which implies $s=\alpha_i(r)\ge \alpha_j(r)$ for each $j\in [N]\setminus \{i\}$. We also have $\alpha^{-1}(s)=r=\alpha_i^{-1}(s)$, and since $\alpha_i^{-1}(s)=r\le \alpha_j^{-1}(s)$ for each $j\in [N]\setminus \{i\}$, we obtain the result.
\end{proof}
\vspace{0.1cm}
\begin{lemma}\label{lem_diss}
    For $i=1,2$, consider the functions $\alpha_i:\mathbb{R}_{\geq0}\to\mathbb{R}_{\geq0}$ and $\sigma_i:\mathbb{R}_{\geq0}\to\mathbb{R}_{\geq0}$, given by
    \begin{subequations}
        \begin{align*}
            \alpha_i(s)&=a_{i} s^{p_i}+b_{i}s^{q_i},\quad a_i,\ b_i>0,\ p_i\in (0,1),\ q_i>1\\
            \sigma_i(s)&=\sum_{j=1}^{n_i}c_{i,j} s^{\eta_{i,j}},\quad c_{i,j},\ \eta_{i,j}>0 \ \ \forall j\in [n_i].%\\
            %\sigma_2(s)&=\sum_{j=1}^{n_2}c_{2,j} s^{\eta_{2,j}},\quad c_{2,j},\ \eta_{2,j}>0 \ \ \forall j\in [n_2]
        \end{align*}
    \end{subequations}
    Let $\gamma_i:=\alpha_i^{-1}\circ\sigma_i$, $c_i:=\max_j c_{i,j}$, and suppose that \eqref{conditionthm2} holds and ${\eta_{1,i}}{\eta_{2,j}}\in [p_2, q_2]$ for all $i\in [n_1], j\in [n_2]$. Then, there exists some $C>0$ such that if $c_2 \max_k c_1^{\eta_{2,k}}< C$, there exists $\hat{\gamma}_1\in \ckp$ and $\hat{\gamma}_2\in\cki$ such that $\hat{\gamma}_i(s)\ge\gamma_i(s)$ and $\hat{\gamma}_1\circ\hat{\gamma}_2(s)<s$ for all $s>0$.\QEDB
\end{lemma}

\begin{proof}
    Let $\delta_i:=\min(a_i, b_i)$ and note that $\alpha_1(s)> \delta_1 s$ for all $s>0$, which means $\alpha_1^{-1}(s)<\frac{1}{\delta_1}s$ for all $s>0$. Hence, for all $s>0$, we have
    \begin{align}        \gamma_1(s)&<\frac{c_1}{\delta_1}\sum_{j=1}^{n_1} s^{\eta_{1,j}}:=\hat{\gamma}_1(s).
    \end{align}
    Since $\sigma_2(s)\le \max_{k}\{n_2c_{2}s^{\eta_{2,k}}\}$, by Lemma \ref{lem_maxinv} we also have
    \begin{equation}
        \sigma_2^{-1}(s)\ge \min_k\{(n_2 c_2)^{-\frac{1}{\eta_{2,k}}}s^\frac{1}{\eta_{2,k}}\},
    \end{equation}
    which implies
    \begin{align}
        &\gamma_2^{-1}(s)=\sigma_2^{-1}\circ\alpha_2(s)\notag\\&\ge \min_k \left\{(n_2 c_2)^{-\frac{1}{\eta_{2,k}}}\left(\delta_2 (s^{p_2}+s^{q_2})\right)^\frac{1}{\eta_{2,k}}\right\}.\label{diss_step1}
    \end{align}
    By Lemma \ref{jensenlemma}, we have $\left(s^{\frac{p_2}{\eta_{2,k}}}+s^{\frac{q_2}{\eta_{2,k}}}\right)^{\eta_{2,k}}\le M_k(s^{p_2}+s^{q_2}),$
    for each $k\in[n_2]$, where $M_k=\max(1, 2^{\eta_{2,k}-1})$. Dividing by $M_k$ and raising both sides to the $\frac{1}{\eta_{2,k}}$-th power, we obtain:
    \begin{equation}
        \left(\frac{1}{M_k}\right)^{\frac{1}{\eta_{2,k}}}\left(s^{\frac{p_2}{\eta_{2,k}}}+s^{\frac{q_2}{\eta_{2,k}}}\right)\le \left( s^{p_2}+s^{q_2}\right)^\frac{1}{\eta_{2,k}}.\label{diss_step2}
    \end{equation}
    By leveraging \eqref{diss_step2}, we can continue from \eqref{diss_step1} to obtain
    \begin{align*}
        \gamma_2^{-1}(s)&\ge \min_k\left\{\left(\frac{\delta_2}{n_2c_2{M}_k}\right)^\frac{1}{\eta_{2,k}}\left(s^{\frac{p_2}{\eta_{2,k}}}+s^{\frac{q_2}{\eta_{2,k}}}\right)\right\}\\&>M (s^{p}+s^{q}):=\hat{\gamma}_2^{-1}(s),
    \end{align*}
   for $s>0$, where $M=\frac12\min_k \left(\frac{\delta_2}{n_2c_2{M}_k}\right)^\frac{1}{\eta_{2,k}}$, $p=\max_k\frac{p_2}{\eta_{2,k}}$, $q=\min_k\frac{q_2}{\eta_{2,k}}$, and the bottom inequality follows from Lemma \ref{lem_sandw}. Since we assume ${\eta_{1,i}}{\eta_{2,j}}\in [p_2, q_2]$ for all $i\in [n_1], j\in [n_2]$, we can again apply Lemma \ref{lem_sandw} to obtain $\hat{\gamma}_1(s)<\frac{c_1 n_1}{\delta_1}(s^p+s^q).$
    We then obtain that $\hat{\gamma}_1(s)<\hat{\gamma}_2^{-1}(s)$ if $\frac{c_1n_1}{\delta_1}<M$ holds, which, by the definition of $M$, can also be expressed as
    \begin{equation*}
       \max_k\left\{\left(\frac{2c_1 n_1}{\delta_1}\right)^{\eta_{2,k}}\frac{n_2c_2{M}_k}{\delta_2}\right\}<1. 
    \end{equation*}
    Since the following holds
    \begin{align*}
        &\max_k\left\{\left(\frac{2c_1 n_1}{\delta_1}\right)^{\eta_{2,k}}\frac{n_2c_2{M}_k}{\delta_2}\right\}\\&\le \left(c_2\max_k c_1^{\eta_{2,k}}\right)\max_j\left\{\left(\frac{2 n_1}{\delta_1}\right)^{\eta_{2,j}}\frac{n_2{M}_j}{\delta_2}\right\},
    \end{align*}
    we can select $C=\min_j\frac{\delta_2}{n_2 M_j}\left(\frac{\delta_1}{2n_1}\right)^{\eta_{2,j}}$ to complete the proof.
\end{proof}
\vspace{0.1cm}
\begin{lemma}\label{holder_ineq}
    Given $s\in\re^n$ and $0<p<q$, the following holds
    \begin{equation*}
        \left(\sum_{i=1}^n |s_i|^q\right)^{\frac{1}{q}}\le\left(\sum_{i=1}^n |s_i|^p\right)^{\frac{1}{p}}\le n^{\frac{1}{p}-\frac{1}{q}}\left(\sum_{i=1}^n |s_i|^q\right)^{\frac{1}{q}}.
    \end{equation*}
    \QEDB
\end{lemma}
\begin{proof}
    See \cite[Prop 6.11]{folland1999real} and \cite[Prop 6.12]{folland1999real}.
\end{proof}
\vspace{0.1cm}
\begin{lemma}\label{lem_decentralized}
    Let $N\in \mathbb{N}$, $p<1$, and $M\in (0, \tcb{\frac{1}{\sqrt{N}}})$.
    Then, there exists $c>0$ such that
    \begin{equation}\label{decent_lem}
        \sum_{i=1}^N \frac{a_i^2+a_i b_i}{|a_i+b_i|^p}\ge c |a|^{2-p},
    \end{equation}
    for all $a,b\in\re^N$ that satisfy $|b|\le M|a|$, {where the summand in \eqref{decent_lem} is understood to equal $0$ when $a_i+b_i=0$.}\QEDB
\end{lemma}
\begin{proof}
\tcb{We assume $N\ge 2$, since if $N=1$ we can take $c=(1-M)^{1-p}$. Consider the function $\zeta:\re^N\to\re^N$ where the $i$-th component is given by $\zeta_i(x)=\frac{x_i}{|x_i|^p}$ and $\zeta_i(0)=0$. Note that the left hand side of \eqref{decent_lem} satisfies
\begin{equation*}
    \sum_{i=1}^N \frac{a_i^2+a_i b_i}{|a_i+b_i|^p}=a^\top \zeta(a+b).
\end{equation*}
Letting $x=a+b$, we can apply Lemma \ref{holder_ineq} to $|x|^2$ and $|\zeta(x)|^2$ to obtain
\begin{subequations}
\begin{align}
    |x|^2&=\sum_{i=1}^N |x_i|^2\le N^\frac{1-p}{2-p}\left(\sum_{i=1}^N |x_i|^{4-2p}\right)^{\frac{1}{2-p}}\label{zxhold1}\\ |\zeta(x)|^2&=\sum_{i=1}^N |\zeta_i(x)|^2\le N^\frac{1}{2-p}\left(\sum_{i=1}^N |x_i|^{4-2p}\right)^{\frac{1-p}{2-p}}.\label{zxhold2}
\end{align}
\end{subequations}
Combining equations \eqref{zxhold1}-\eqref{zxhold2} yields
\begin{equation}\label{hold1}
    |x||\zeta(x)|\le \sqrt{ N\sum_{i=1}^N |x_i|^{4-2p}}\le \sqrt{N}\sum_{i=1}^N |x_i|^{2-p}.
\end{equation}
Since $|b|\le M|a|$ and $M<1$, we have
\begin{align*}
    |(1-M^2)a-x|&=\sqrt{M^2|x|^2+(1-M^2)(|b|^2-M^2|a|^2)}\\&\le M|x|.
\end{align*}
We can combine this estimate with \eqref{hold1} to obtain
\begin{align}
    (1-M^2)a^\top \zeta(x)&=x^\top\zeta(x)+((1-M^2)a-x)^\top\zeta(x)\notag \\
    &\ge x^\top \zeta(x)-M|x||\zeta(x)|\notag\\
    &\ge \left(1-M\sqrt{N} \right)\sum_{i=1}^N |x_i|^{2-p}\label{est1}.
\end{align}
A comparison of the $\mathcal{L}_{2-p}$ and $\mathcal{L}_2$ norms, combined with reverse triangle inequality, gives
\begin{align}
    \sum_{i=1}^N |x_i|^{2-p}&\ge N^{\min\{0, \frac12 p\}}|x|^{2-p}\notag\\&\ge N^{\min\{0, \frac12 p\}}(1-M)^{2-p}|a|^{2-p},\label{comp_rev}
\end{align}
We then substitute \eqref{comp_rev} into \eqref{est1} to obtain the result, with
\begin{equation}
    c:=\frac{1-M\sqrt{N}}{1-M^2}N^{\min\{0, \frac12 p\}}(1-M)^{2-p}>0,
\end{equation}
which only depends on $M, N, p$.
}\end{proof}

\bibliographystyle{IEEEtran}
\bibliography{Biblio.bib}

\end{document}